\documentclass[a4paper,11pt,oneside,reqno]{amsart}
\usepackage[utf8]{inputenc}
\usepackage{amsmath,amsthm,amsfonts,latexsym,amssymb,bm,enumerate}
\usepackage{ae}
\usepackage{cite}
\usepackage{float}
\usepackage{lmodern}
\usepackage[T1]{fontenc}

\usepackage{color}

\usepackage[colorlinks=true]{hyperref}

\definecolor{dark-red}{rgb}{.54,.0,.0}
\definecolor{dark-green}{rgb}{.0,.4,.0}
\definecolor{dark-blue}{rgb}{.04,.04,.4}

\hypersetup{linkcolor=dark-red, urlcolor=dark-blue, citecolor=dark-green}

\usepackage[dvips]{graphicx}
\usepackage{psfrag}
\DeclareGraphicsExtensions{.eps,.art,.ART,.ps}

\usepackage{rotating}

\newcounter{mnotecount}[section]

\usepackage[dvips]{graphicx}
\usepackage{psfrag}
\DeclareGraphicsExtensions{.eps,.art,.ART,.ps}

\DeclareFontFamily{OT1}{pzc}{}
\DeclareFontShape{OT1}{pzc}{m}{it}%
              {<-> s * pzcmi8t}{}
\DeclareMathAlphabet{\mathpzc}{OT1}{pzc}%
                                {m}{it}

\newtheorem{Thm}{Theorem}[section]

\newtheorem{Lem}[Thm]{Lemma}
\newtheorem{Cor}[Thm]{Corollary}
\newtheorem{Prop}[Thm]{Proposition}

\newtheorem{Rmk}[Thm]{Remark}
\newtheorem{Def}[Thm]{Definition}
\newtheorem{Problem}[Thm]{Problem}

\newcommand{\cal}{\mathcal}

\renewcommand{\Psi}{\rho}

\newcommand{\eps}{\varepsilon}

\newcommand{\vlinha}{\frac{e^2}{r}+\frac{\Lambda}{3}r^3-\omega}
\newcommand{\vlinhazz}{\frac{e^2}{r}+\frac{\Lambda}{3}r^3-\omega_0}
\newcommand{\vlinhaz}{\frac{e^2}{{r}_0(0)}+\frac{\Lambda}{3}{r}_0^3(0)-\omega_0}
\newcommand{\muu}{1-\frac{2\omega}{r}+\frac{e^2}{r^2}-\frac{\Lambda}{3}r^2}
\newcommand{\muuz}{1-\frac{2\omega_0}{\tilde{r}_0}+\frac{e^2}{\tilde{r}_0^2}-\frac{\Lambda}{3}\tilde{r}_0^2}
\newcommand{\truemu}{\frac{2\omega}{r}-\frac{e^2}{r^2}+\frac{\Lambda}{3}r^2}
\newcommand{\mysigma}{{\cal P}}
\renewcommand{\omega}{\varpi}
\newcommand{\myumax}{U}
\newcommand{\myU}{U'}
\newcommand{\Aw}{A_{\omega}}
\newcommand{\ru}{\partial_r(1-\mu)(r_+,\varpi_0)}
\newcommand{\ckrm}{\check r_-}
\newcommand{\ckrp}{\check r_+}

\newcommand{\cg}{\Gamma}
\newcommand{\vgr}{v_{\ckrm}}
\newcommand{\ugr}{u_{\ckrm}}
\newcommand{\ugrp}{u_{\ckrp}}
\newcommand{\vgrp}{v_{\ckrp}}
\newcommand{\ug}{u_r}
\newcommand{\vg}{v_r}
\newcommand{\ugs}{u_s}
\newcommand{\vgs}{v_s}
\newcommand{\vgsd}{v_{s_2}}
\newcommand{\ugz}{u_{r_0}}
\newcommand{\vgz}{v_{r_0}}
\newcommand{\ckr}{\check{r}}
\newcommand{\uckr}{u_{\ckr}}
\newcommand{\vckr}{v_{\ckr}}
\newcommand{\udz}{u_{r_+-\delta}}
\newcommand{\vdz}{v_{r_+-\delta}}
\newcommand{\udp}{u'_{r_+-\delta}}
\newcommand{\gam}{\gamma}
\newcommand{\ugam}{u_{\gam}}
\newcommand{\vgam}{v_{\gam}}

\begin{document}

\newcounter{enumii_saved}

\title[Global uniqueness with a cosmological constant - Part 1]{On the global uniqueness for the Einstein-Maxwell-scalar field system with a cosmological constant \\ \vspace{.2cm}
{\small Part 1. Well posedness and breakdown criterion}}

\author{Jo\~ao L.~Costa}
\author{Pedro M.~Gir\~ao}
\author{Jos\'{e} Nat\'{a}rio}
\author{Jorge Drumond Silva}

\address{Jo\~ao L.~Costa: 
ISCTE - Instituto Universitário de Lisboa, Portugal
and 
Center for Mathematical Analysis, Geometry and Dynamical Systems,
Instituto Superior T\'ecnico, Universidade de Lisboa, Portugal
}
\email{jlca@iscte.pt}

\address{Pedro M.~Gir\~ao, Jos\'{e} Nat\'{a}rio and Jorge Drumond Silva: 
Center for Mathematical Analysis, Geometry and Dynamical Systems,
Instituto Superior T\'ecnico, Universidade de Lisboa, Portugal}
\email{pgirao@math.ist.utl.pt}
\email{jnatar@math.ist.utl.pt}
\email{jsilva@math.ist.utl.pt}

\thanks{Partially funded by FCT/Portugal through project PEst-OE/EEI/LA0009/2013.
P.~Girão and J.~Silva were also partially funded by FCT/Portugal through grants PTDC/MAT114397/2009 and UTA$\underline{\ }$CMU/MAT/0007/2009.}

\maketitle

\begin{center}
{\bf Abstract}
\end{center}

This paper is the first part of a trilogy dedicated to the following problem: 
given spherically symmetric characteristic
initial data for the Einstein-Maxwell-scalar
field system with a cosmological constant $\Lambda$,
with the data on the outgoing initial null
hypersurface given by a subextremal Reissner-Nordstr\"{o}m black
hole event horizon, study the future extendibility of the corresponding
maximal globally hyperbolic development (MGHD) as a  ``suitably regular'' Lorentzian manifold. 

In this first part we establish well posedness of the Einstein equations for characteristic data satisfying the minimal regularity conditions
leading to classical solutions. We also identify the appropriate notion
of maximal solution, from which the construction of the corresponding MGHD follows, and determine breakdown criteria. This is the unavoidable starting point of the analysis;
our main results will depend on the detailed understanding of these fundamentals.

In the second part of this series~\cite{relIst2} we study the stability of the radius function at the Cauchy horizon. In the third 
and final paper~\cite{relIst3} we show that, depending on the decay rate of the initial data, mass inflation may or may not occur; in fact, it
is even possible to have (non-isometric) extensions of the spacetime across the Cauchy horizon as classical solutions of the Einstein equations. 

\newpage

{
\setcounter{tocdepth}{1}
\tableofcontents
}

\section{Introduction}
\subsection{The problem and its context in Mathematical Physics}
In general relativity, the question of determinism reduces to the mathematical
problem of global uniqueness for solutions of the Einstein equations. This is a consequence of
the central role played by the well posedness results for the Cauchy
problem, first established in the seminal work of Choquet-Bruhat~\cite{Choquet:1952} and Choquet-Bruhat
and Geroch~\cite{Choquet:1969}.

The global uniqueness issue lies beyond the primary difficulty concerning the diffeomorphism
invariance of the equations. It can be conveniently framed by introducing the concept of a
{\em maximal globally hyperbolic development} (MGHD), informally the largest
Lorentzian manifold $({\cal M},g)$ determined, via Einstein's equations, by the initial
data.\footnote{For the precise definitions and a modern account of the Cauchy problem in general
relativity, see~\cite{RingstromCauchy}.} Existence and uniqueness (up to diffeomorphism) of a MGHD
is then the outcome of the well posedness results mentioned above. Remarkably, the study of exact solutions
reveals the possibility of constructing non-isometric extensions of the MGHD. So,
we see that uniqueness of the MGHD is, in fact, a local result: a na\"\i ve version
of  global uniqueness, aiming for full generality, simply fails.

The boundary of the MGHD as a proper submanifold of a larger spacetime,
when such an embedding is possible, is known as the Cauchy horizon. An important class of
examples is provided by the three parameter
[$(M,4\pi e,\Lambda)=$ (mass, charge, cosmological constant)] family of solutions to the Einstein-Maxwell equations
known as the Reissner-Nordstr\"{o}m family.\footnote{Throughout this work we will simply use ``Reissner-Nordst\"om'' to mean any of the anti-de Sitter ($\Lambda<0$), the asymptotically flat ($\Lambda=0$), or the de
Sitter ($\Lambda>0$) Reissner-Nordstr\"{o}m solutions.}
Within a relevant parameter range, and given appropriate initial data, all solutions of this family contain a Cauchy
horizon inside the black hole region.
In fact, all observers (causal curves) entering the black hole region will reach the Cauchy horizon in finite proper
time and will ``safely'' cross such horizon. However, their future history beyond that point (which might continue indefinitely)
cannot be uniquely determined from the initial data (see Figure~\ref{fig23}).

\begin{figure}[h!]
\begin{center}
\psfrag{m}{\tiny MGHD}
\psfrag{e}{\tiny event horizon}
\psfrag{h}{\tiny Cauchy horizon}
\psfrag{s}{\tiny singularity $r=0$}
\psfrag{c}{\tiny causal curve}
\includegraphics[scale=.9]{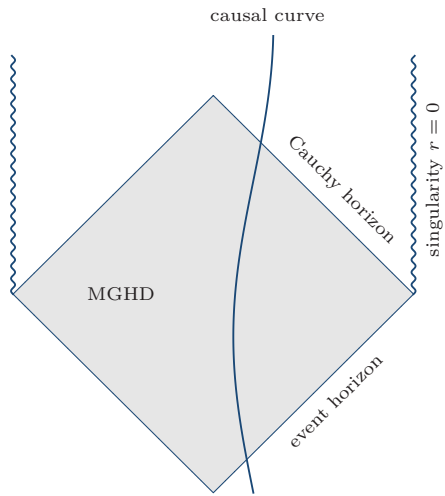}
\end{center}
\vspace{-1.5cm}
\caption{Penrose diagram for part of the Reissner-Nordstr\"{o}m solution.}\label{fig23}
\end{figure}

Some early important insights concerning global uniqueness came from an heuristic argument,
complemented with numerical experiments, by Simpson and Penrose~\cite{SimpsonInternal}, who provided
evidence in favor of the instability of the Cauchy horizon of some Reissner-Nordstr\"{o}m solutions with $\Lambda=0$. Later,
Israel and Poisson~\cite{IsraelPoisson} identified the blow up of a scalar invariant known as the Hawking mass as the source of the
instability, in a process known as {\em mass inflation}.
In view of these developments, the expectation became that generic perturbations of such solutions
should turn the Cauchy horizon into a singularity beyond which the spacetime could not be continued in any meaningful way.
This reinstated the belief in global uniqueness as a generic property of reasonable initial value problems
for the Einstein equations, an idea substantiated in Penrose's {\em strong cosmic censorship conjecture},\footnote{For
those wondering about the choice of such a ``baroque name''~\cite{ChruscielSCC}
and not familiar with its history, it might help to note that it is related to
a sibling conjecture, the {\em weak cosmic censorship}, which forbids the generic
existence of {\em naked singularities}.}
see~\cite{PenroseSingularities}, \cite{ChruscielSCC}, \cite{ChristodoulouGlobalnew} and~\cite{DafermosBlack}.

In this series of papers we will study the relation between Cauchy horizon stability and global uniqueness by
considering the full non-linear evolution of an appropriate Einstein-matter
system to the future of the event horizon (the boundary of the black hole region).
Arguably, the simplest formulation of the global uniqueness question within this framework takes the following form:

\begin{Problem}
\label{problema}
  Given spherically symmetric characteristic
  initial data for the Einstein-Maxwell-scalar
  field system with a cosmological constant $\Lambda$,
  with the data on the outgoing initial null hypersurface
  given by a (complete) subextremal Reissner-Nordstr\"{o}m black
  hole event horizon with non-vanishing charge $4\pi e$, and the remaining data otherwise free, study the 
  future extendibility of the corresponding MGHD as a  ``suitably regular'' Lorentzian manifold.
\end{Problem}

 Let us take a moment to discuss the choices made. We take characteristic initial data due to the null geometry of event horizons.
Spherical symmetry is compatible with any choice of sign for the cosmological constant, and at the same time allows us to reduce the Einstein
equations to a $1+1$ evolution problem, which is considerably simpler than its higher dimensional counterpart.
The matter model, a self-gravitating real massless scalar field, provides the simplest
non-pathological Einstein-matter system with dynamical degrees of freedom in spherical symmetry,
and admits the entire Reissner-Nordstr\"{o}m family as solutions. Moreover, it exhibits a wavelike behavior reminiscent of the
general Einstein vacuum equations. A non-vanishing charge parameter is needed to exclude the Schwarzschild subfamily, 
whose solutions do not contain a Cauchy horizon to start with.
 
A symmetric model is, a priori, non-generic. Nonetheless, there exists a relation between the spherically symmetric self-gravitating scalar field model
and the problem of vacuum collapse without symmetries, which
has been particularly enlightening and fruitful (see \cite{Christodoulou:2008}).
In fact, recently \cite{LukWeak} explored this relation to obtain the first 
promising steps towards the understanding of the 
stability of Cauchy horizons without symmetry assumptions. We also note that, in
our framework, the charge is topological, hence non-dynamical. 
We refer the reader to the discussion 
in~\cite{KommemiGlobal}, where preliminary results for a more realistic matter model are obtained.

The question of the regularity of the possible extensions is of paramount importance and  was
not precisely formulated above. Obtaining a precise formulation is, in fact, one of the challenges of the
problem. The definition of  a ``suitably regular'' extension should, of course, exclude artificial extensions, like, for instance,
taking the disconnected union of the MGHD with another Lorentzian manifold. Having done that, we are still left with several
possibilities. Let us start by considering standard regularity requirements. We will concentrate
on the regularity of the metric, but the regularity requirements for the matter fields,
in our case the scalar field, must also be discussed.

{\em Inextendibility of the metric in $C^2$}: This is motivated by the fact that the Einstein equations are
of second order. Moreover, causality theory (which is extensively used in the proof of the existence and uniqueness of a MGHD) usually assumes 
the metric to be $C^2$ (see~\cite{ChruscielGrant}). 
It was this criterion that was used
by Ringstr\"om in his seminal work on strong cosmic censorship 
for Gowdy symmetry \cite{RingstromGowdy}. Clearly, the existence of $C^2$ extensions would provide
strong evidence against any reasonable form of global uniqueness, unless, for some unlikely and mysterious motive,
the spacetime structure would turn out to be uniquely determined beyond the corresponding Cauchy horizons
(see the discussion in~\cite{ChruscielSCC}). On the other hand, $C^2$-inextendibility does not necessarily provide a
compelling argument in favor of global uniqueness, since there are relevant solutions of the Einstein equations whose regularity is
well below this threshold. For instance, the impulsive wave solutions of Luk and Rodnianski~\cite{LukLocal, LukLocal2},
generalizing the plane gravitational waves of Khan and Penrose~\cite{KhanPenrose} and Szekeres~\cite{SzekeresColliding}, admit a Dirac delta singularity in the curvature along
a null hypersurface, which in no way should be considered as a terminal boundary of spacetime. Another remarkable
example is provided
by  Ori's solution~\cite{OriInner}, where the effect of the pointwise blow-up of the Kretschmann scalar at the Cauchy
horizon, corresponding to infinite tidal forces there, does not necessarily lead to the ``destruction'' of an
observer crossing the horizon; more precisely, a double integral of the Kretschmann scalar remains finite.
Other solutions to the Einstein equations with lower regularity than $C^2$ have been studied in the literature (see for instance \cite{ChristodoulouBounded, LeFloch2007, LeFloch2014}
and references therein). In fact, we will see in this series of papers that there even exist classical solutions of the Einstein equations which
are not necessarily $C^2$.

{\em Inextendibility of the metric in $C^0$}: This was introduced by Christodoulou in his original formulation of the
strong cosmic censorship conjecture~\cite{ChristodoulouGlobalnew}. Its validity would provide overwhelming
evidence in favor of the deterministic nature of general relativity.
It holds for the Schwarzschild solution, and indeed for the generic asymptotically flat solutions of the Einstein-massless scalar field system studied by Christodoulou\footnote{Christodoulou showed that $r=0$ at the future boundary of the black hole region of a generic solution. This suffices to show that there are no spherically symmetric $C^0$ extensions, and it is widely believed (but we are not aware of a formal proof) that in fact there are no $C^0$ extensions whatsoever. In the present work we will only consider the question of existence of spherically symmetric extensions.} in \cite{ChristodoulouNaked}.

It turns out that the two regularity requirements above are insufficient
to capture the full richness of the subject.
In~\cite{Dafermos1}, Dafermos solved many instances of Problem~\ref{problema}
for the $\Lambda=0$ case. His findings were quite remarkable: first, under
the hypothesis of Problem~\ref{problema}, the MGHD can always be extended in a $C^0$ 
manner; second, if the initial data is sufficiently subextremal and if, in addition, we
assume an appropriate decay
for the free initial data (more precisely, an appropriate decay for the ingoing derivative of the scalar field),
then mass inflation occurs.
Since the Hawking mass is a geometric invariant involving first derivatives of the metric, its blow up excludes the
existence of (spherically symmetric) $C^1$ extensions. 

These results led Christodoulou~\cite{Christodoulou:2008} to drop his earlier $C^0$ formulation in favor
of a stronger regularity requirement: the natural inextendibility criterion, from the PDE point of view, would be to exclude extensions which
are solutions of the Einstein equations, as opposed to arbitrary extensions.

{\em Christodoulou-Chru\'sciel inextendibility criterion --
  inextendibility as a Lo\-ren\-tzi\-an manifold with Christoffel
symbols in $L^2_{\rm loc}$}:
This is enough to ensure that no extension satisfying the Einstein equations, even in a weak sense,
is possible.  A somewhat different kind of reasoning had already led Chru\'sciel, in~\cite{ChruscielSCC}, to consider
the (slightly stronger) regularity conditions $g\in H^1_{\rm loc}$  and $g^{-1}\in L^{\infty}_{\rm loc}$, as being sufficiently low for the well-posedness of the Einstein equations not to hold. 
Indeed, to our knowledge, no well posedness results exist, or are expected to exist, at this level of regularity: the state of the art concerning the Cauchy problem requires
square integrable curvature~\cite{KlainermanL2}. Therefore, the consequences
of a potential failure of the Christodoulou-Chru\'sciel
inextendibility criterion might not be as definite as the consequences of its success: such failure may just mean that the 
appropriate regularity for the extension criterion must be made higher.

From the techniques in~\cite{DafermosBlack}, one can easily conclude that, under the conditions
leading to mass inflation identified in~\cite{Dafermos1}, the Christodoulou-Chru\'sciel inextendibility criterion holds (in spherical symmetry).

\bigskip

We now motivate the introduction of a cosmological constant~$\Lambda$ in Problem~\ref{problema}.
It is well known that it plays a fundamental role in modern physics: in cosmology, for instance, a positive $\Lambda$ provides the simplest mechanism to model dark energy,
 while in high energy physics a great deal of attention has been devoted to solutions of
the Einstein equations with $\Lambda<0$. From a purely mathematical point of view, the introduction of this zeroth order term in the
Einstein equations leads
to a whole new range of rich geometrical structures and dynamical behaviors.

Nonetheless, the relevance of considering a cosmological constant in Problem~\ref{problema} is not obvious a priori.
This problem concerns what happens in the interior of a black hole region, whereas the effects of the
cosmological constant are known
to be specially relevant at large scales, in the exterior regions.  It turns out that the question is far more interesting than
originally anticipated. In fact, there was a considerable amount of activity during the 90s concerning the
strong cosmic censorship conjecture with a positive cosmological constant~\cite{BradyCauchy, BradyCauchyMass}. Based on heuristic arguments and perturbative analyses, the initial
expectation became that, close to extremality, mass inflation should fail, i.e.~the Hawking mass should remain bounded, and that
$C^2$ extensions of the spacetime beyond the corresponding Cauchy horizon should exist. In conclusion, the strong cosmic censorship conjecture
was believed to fail for $\Lambda>0$. Later, a perturbative analysis based on a more sophisticated model showed
that the previous results neglected the effects of backscattering; these, when properly taken into account,
would lead to the blow up of curvature scalars. The question of mass inflation was not addressed in the context of this new model.
The main thesis was clear from the title under which these results were presented:
``Cosmic censorship: As strong as ever''~\cite{BradyCosmic} (at least in its $C^2$ inextendibility formulation).

This remained the dominant point of view until, recently,
Dafermos suggested that the
expected curvature blow up might not be related to the blow up of the mass, which, close to extremality, should remain bounded.
The consequences of such {\em no mass inflation} scenario should be considerable and led
Dafermos to conjecture~\cite[Conjecture~4]{DafermosBlack} that
the Christodoulou-Chru\'sciel
inextendibility criterion would fail for positive $\Lambda$. We will show that this is in fact the case
in the context of Problem~\ref{problema}, for any sign of $\Lambda$ (including the $\Lambda=0$ case treated in \cite{Dafermos1}), provided that the free data decays sufficiently fast (see Part~3). It should be noted, however, that we are considering pure Reissner-Nordstr\"{o}m data on the event horizon, which in general will not arise from gravitational collapse; the case $\Lambda > 0$ is special in that this data is expected to be approached exponentially fast (see Section~\ref{section1.3}).

The question of strong cosmic censorship for $\Lambda < 0$ has received less attention in the literature, but it is expected that the mass inflation scenario should hold.

\bigskip

This paper is the first part of a trilogy devoted to the study of Problem~\ref{problema}.
We study the relation between the spherically symmetric
Einstein-Maxwell-scalar field equations with a cosmological constant and the first order PDE system~\eqref{r_u}$-$\eqref{kappa_at_u}, for the quantities~\eqref{nu_0}$-$\eqref{kappa_0}.
We establish its well posedness under the minimal regularity conditions
leading to classical solutions. We also identify the appropriate notion
of maximal solution, from which
the construction of the corresponding MGHD immediately follows, and determine breakdown criteria. In broad terms, these results are widely expected, although
we are unaware of a published proof.
They are also the unavoidable starting point of the analysis; 
our main results will depend on a detailed understanding of these fundamentals, and
seemingly minor improvements here will be of paramount importance later on.

In the second paper, we generalize the results of Dafermos concerning the stability of the radius function at the Cauchy horizon
with the inclusion of a cosmological constant (of any sign). More precisely, we show that the radius function is 
bounded away from zero. This has the remarkable consequence of allowing the construction of continuous extensions of the metric.
The introduction of the cosmological constant requires a considerable deviation from the original strategy developed by Dafermos,
the main difficulties stemming from the $\Lambda>0$ case. We develop a unified framework that avoids dividing the solution spacetime into red-shift
and blue-shift regions, but instead focuses on the level sets of the radius function, without having to distinguish between the different signs for the cosmological constant.

In the third paper of the series, we analyze the mass inflation scenarios. We start by revisiting Dafermos'
strategy to establish mass inflation, which naturally generalizes to the case of a non-vanishing cosmological constant.
Then we show that under appropriate circumstances (in particular close to extremality), the mass remains bounded. As a consequence,
the Christodoulou-Chru\'sciel inextendibility criterion fails, i.e.\ we can construct extensions of spacetime, across the Cauchy horizon, with continuous metric,
square integrable connection coefficients, and scalar field in $H^1_{\rm loc}$. Finally we prove that, under
slightly stronger conditions, we can bound the gradient of the scalar field. This allows us to construct
(non-isometric) $C^1$ extensions of the metric, which correspond to classical solutions of the Einstein-Maxwell-scalar field equations.
To the best of our knowledge, these are the first results where the generic existence of extensions as solutions is established.

\subsection{The main results}

We summarize the  main results of this trilogy in 
\begin{Thm}
\label{thmMain}
Consider, as a reference solution, a subextremal element of the Reissner-Nordstr\"{o}m family with mass $M$, non-vanishing charge $4\pi e$ and cosmological constant $\Lambda$, and let $r_+>0$ be the event horizon radius. Given $0 < U < r_+$, there exists a unique maximal development of the characteristic initial value problem for the spherically symmetric Einstein-Maxwell-scalar field system with metric
\[
g=-\Omega^2(u,v)\,dudv+r^2(u,v)\,\sigma_{\mathbb{S}^2} \label{metric-g}
\]
and scalar field $\phi$, defined on a past set ${\cal P} \subset [0,U] \times [0,+\infty[$, determined by Reissner-Nordstr\"{o}m initial data on the event horizon $\{u=0\}$ and free data on the incoming null direction $\{v=0\}$:
\begin{enumerate}[i)]
\item $r(0,v) = r_+$;
\item $\phi(0,v) = 0$;
\vspace{-3mm}
\item $\Omega^2(0,v) = 4 e^{-\,\frac{2}{r_+^2}\bigl(\frac{e^2}{{r_+}}+\frac{\Lambda}{3}r_+^3-M\bigr)\,v}$;
\item $r(u,0) = r_+-u$;
\item $\phi(\,\cdot\,,0)\in C^1([0,U])$ with $\partial_u \phi(0,0)=0$.
\end{enumerate}

Moreover:
\begin{enumerate}[{\rm (1)}]
\item  {\em Stability of the radius function at the Cauchy horizon.} There exists $U>0$ such that
$$[0,U]\times[0,\infty[\,\subset {\cal P},$$
and $r_0>0$ for which
$$r(u,v)>r_0,\ {\rm for\ all}\ (u,v)\in[0,U]\times[0,\infty[.$$
Consequently, $({\cal M},g,\phi)$ extends, across the Cauchy horizon $\{v=\infty\}$, to $(\hat {\cal M},\hat g,\hat \phi)$, with $\hat g$ and $\hat \phi$ in~$C^0$.
\end{enumerate}
Let\/ $1<\rho<\infty$ be the ratio between the surface gravities of the
Cauchy horizon and of the event horizon of the reference solution.
\begin{enumerate}[{\rm (1)}]\addtocounter{enumi}{+1}
\item {\em Mass inflation.} If $\rho>2$ and
 \begin{equation*}
\partial_u\phi(u,0)\geq cu^s,\ {\it for\ some}\/\
0<s<\frac\rho 2-1,
\end{equation*}
then the renormalized Hawking mass $\varpi$ (see \eqref{bar_rafaeli}) satisfies
\begin{equation*}
\lim_{v\rightarrow\infty}\varpi(u,v)=\infty,\ {\it for\ each}\/\  0<u\leq U.
\end{equation*}
In particular, no (spherically symmetric) $C^1$ extensions across the Cauchy horizon exist.
\item {\em No mass inflation.}  If 
\begin{equation*}
|\partial_u\phi(u,0)|\leq cu^s,\ {\rm for\ some}\ s>\frac{7\Psi}9-1,
\end{equation*}
then there exists $C>0$ such that
$$
\left|\varpi(u,v)\right|<C,\ {\rm for\ all}\ (u,v)\in[0,U]\times[0,\infty[.
$$
Consequently, the Christodoulou-Chru\'sciel inextendibility criterion fails, i.e.\
$({\cal M},g,\phi)$ extends across the Cauchy horizon to $(\hat {\cal M},\hat g,\hat \phi)$, with $\hat g$ and $\hat \phi$ in~$C^0$,
Christoffel symbols $\hat\Gamma$ in $L^2_{\rm loc}$, and $\hat\phi$ in $H^1_{\rm loc}$.
\item {\em Classical extensions.} If 
\begin{equation*}
|\partial_u\phi(u,0)|\leq cu^s,\ {\rm for\ some}\ s>\frac{13\Psi}9-1,
\end{equation*}
then $({\cal M},g,\phi)$ extends (in a non-unique way), across the Cauchy horizon,
to a spherically symmetric (classical) solution $(\hat {\cal M},\hat g,\hat \phi)$ of the Einstein-Maxwell-scalar field system with
cosmological constant $\Lambda$, with $\hat g$ and $\hat\phi$ in $C^1$.
The Kretschmann scalar is uniformly bounded.
\end{enumerate}
\end{Thm}

In case~(4), and when $\partial_u\phi(\,\cdot\,,0)$ is just $C^0$, the metric does not have to be $C^2$ 
(see Part~3).

In fact, we can say more about the behavior of the solution at the Cauchy horizon.
For example, suppose that $cu^s\leq\partial_u\phi(u,0)\leq Cu^s$.
Then the behavior of the solution  depends on the value of $s$ as described in Figure~\ref{fig24}.
 We refer to Part~3 for more details.

\newpage

\begin{figure}[h!]
\begin{center}
\psfrag{b}{\tiny $\!\frac 97$}
\psfrag{u}{\tiny $\!\frac 92$}
\psfrag{c}{\tiny $2$}
\psfrag{d}{\tiny $3$}
\psfrag{e}{\tiny $4$}
\psfrag{z}{\tiny $\!\!1$}
\psfrag{g}{\tiny $\!\!\!\frac 49$}
\psfrag{h}{\tiny $\!\!1$}
\psfrag{i}{\tiny $\!\!2$}
\psfrag{j}{\tiny $\!\!3$}
\psfrag{k}{\tiny $s=\frac \Psi 2-1$}
\psfrag{m}{\tiny $s=\Psi-1$}
\psfrag{l}{\tiny $s=\frac{7\Psi}9-1$}
\psfrag{n}{\tiny $s=\frac{13\Psi}9-1$}
\psfrag{r}{\tiny $\Psi$}
\psfrag{s}{\tiny $s$}
\psfrag{o}{\tiny $1$}
\psfrag{p}{\tiny $\!\!0$}
\includegraphics[scale=.7]{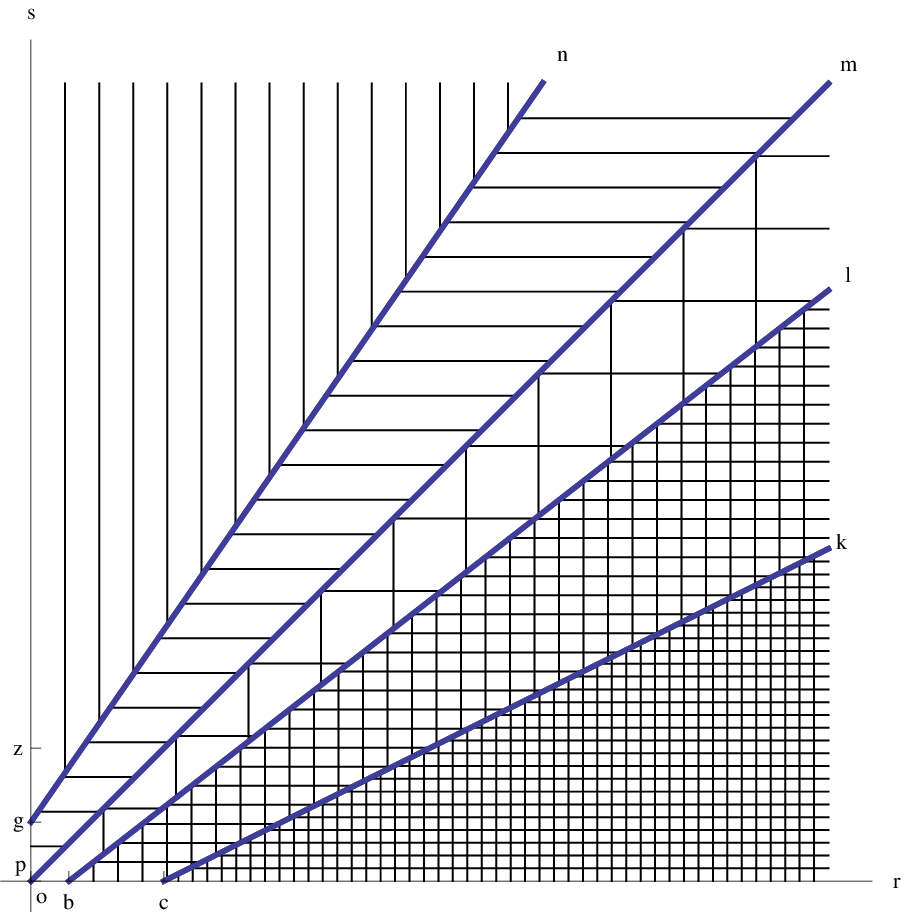}
\end{center}
\begin{center}
\psfrag{e}{\tiny \!\!\!\!\!\!\!\!\!\!\!\!\!\!\!\!\!mass inflation}
\psfrag{u}{\tiny \!\!\!\!\!\!\!\!\!\!\!\!\!\!\!\!\!\!\!\!\!mass inflation or}
\psfrag{b}{\tiny \!\!\!\!\!\!\!\!\!\!\!\!\!\!\!\!\!\!\!\!\!no mass inflation}
\psfrag{a}{\tiny \!\!\!\!\!\!\!\!\!\!\!\!\!\!\!\!\!\!\!\!\!no mass inflation}
\psfrag{c}{\tiny \!\!\!\!\!\!\!\!\!\!\!\!\!\!\!\!\!\!\!no mass inflation}
\psfrag{d}{\tiny \!\!\!\!\!\!\!\!\!\!\!\!\!\!\!\!\!\!\!\!\!\!\!\!\!\!\!\!Kretschmann bounded}
\psfrag{f}{\tiny \!\!\!\!\!\!\!\!\!\!\!\!\!\!\!\!\!\!\!\!\!\!\!\!\!\!\!\!\!\!\!\!Kretschmann unbounded}
\psfrag{v}{\tiny \!\!\!\!\!\!\!\!\!\!\!\!\!\!\!\!\!\!\!\!\!\!\!\!\!\!\!\!\!\!\!\!Kretschmann unbounded}
\psfrag{g}{\tiny \!\!\!\!\!\!\!\!\!\!\!\!\!\!\!\!\!\!\!\!\!smooth extension}
\psfrag{h}{\tiny \!\!\!\!\!\!\!\!\!\!\!\!\!\!\!\!\!\!\!\!\!\!\!\!\!\!\!\!\!beyond Cauchy horizon}
\includegraphics[scale=1]{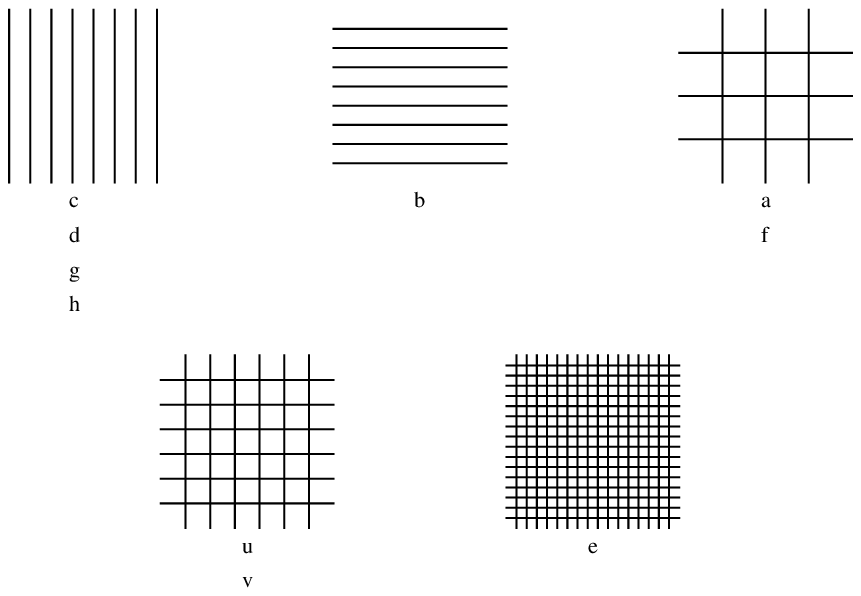}
\end{center}
\vspace{-.5cm}
\caption{Behavior of the solution at the Cauchy horizon.}\label{fig24}
\end{figure}

\subsection{The Strong Cosmic Censorship Conjecture}\label{section1.3}

Strictly speaking, our results do not apply directly to the strong cosmic censorship conjecture since
the data considered on the event horizon does not arise from the gravitational
collapse of generic spacelike initial data. This data is taken from a stationary black hole, and the dynamical features
of our solutions come from the ingoing free data $\phi(\,\cdot\,,0)$.

To strengthen the connection with the strong cosmic censorship conjecture in the asymptotically flat case $\Lambda=0$, 
Dafermos \cite{Dafermos2} extended his original analysis of \cite{Dafermos1} by considering
data along the event horizon with the widely expected scalar field behavior
prescribed by Price's law (see~\cite{PriceNonshperical}):
a polynomial decay, in the coordinate system of Theorem~\ref{thmMain}, for $\partial_v\phi$ 
along the event horizon.  Assuming such decay as an upper bound, he showed that the radius function remains
positive at the Cauchy horizon. Moreover, he showed that if the field satisfies an appropriate lower bound, consistent with Price's law,
then mass inflation occurs throughout the entire subextremal range $1<\Psi<\infty$, and the Christodoulou-Chru\'{s}ciel inextendibility criterion holds~\cite{DafermosBlack}.
The relation with the strong cosmic censorship conjecture was later reinforced in~\cite{DafermosProof}, where Dafermos and Rodnianski proved that the polynomial upper bound for the field indeed holds for
large classes of asymptotically flat data. However, the generic validity of the lower bound, under which mass inflation was established in~\cite{Dafermos2}, remains an open problem.

For $\Lambda>0$, it is widely expected (see~\cite{DafermosWavedeSitter, Vasy2014} and references therein) that
the corresponding Price law should guarantee exponential decay of the scalar field along the event horizon, 
that is, an exponential approach to the data considered here.
If we take into account the fact that near extremality $\Psi$ is approximately equal to $1$, and so the requirements in $(3)$ and $(4)$ of Theorem~\ref{thmMain}
are satisfied with very mild restriction on the behavior of the free data, then it seems plausible that the no mass inflation scenario, and the existence of
regular extensions, should remain valid when $\Psi$ is close to $1$, and appropriately decaying initial data is considered.
We will address this question in a forthcoming paper.

As mentioned above, for $\Lambda=0$ the expected polynomial decay to
stationarity is enough to exclude the no mass inflation scenario and its surprising consequences. When $\Lambda<0$,
the decay of the field along the event horizon is expected to be slower than polynomial\footnote{Interestingly, the
first spherical harmonic mode of the field decays exponentially along the event horizon; thus, strictly speaking, the decay
is exponential in spherical symmetry.}
 (see~\cite{HolzegelDecay} and~\cite{HolzegelQuasimodes}); moreover,
 the process of gravitational collapse is far less understood.
Thus, one may argue that,
although valid for all signs of cosmological constant, points $(3)$ and $(4)$ of
Theorem~\ref{thmMain} only provide evidence
for the failure of the strong cosmic censorship conjecture in the case where $\Lambda$ is positive.
Nonetheless, the techniques used in our work
and the results we obtain suggest that any difference concerning the
stability of the Cauchy horizon when the cosmological constant $\Lambda$ changes sign
should be due to the behavior of the field under the Price law along the event horizon.

\vspace{3mm}

\noindent {\bf Acknowledgments.} J.~Costa thanks P.~Chru\'{s}ciel and M.~Dafermos for useful discussions concerning the strong cosmic censorship conjecture.

\section{An overview of Part~1}

The structure of the present paper is as follows. We consider a double characteristic initial value problem for the spherically symmetric Einstein-Maxwell-scalar field equations with a cosmological constant $\Lambda$.
This consists of a system of four second order Einstein equations: a wave equation~\eqref{wave_r} for the physical radial coordinate~$r$, a wave equation~\eqref{wave_phi} for the field $\phi$, the
Raychaudhuri equation~\eqref{r_uu} in the incoming null direction $u$, and the Raychaudhuri equation~\eqref{r_vv} in the outgoing null direction $v$; and also a wave equation~\eqref{wave_Omega}
for the metric coefficient $\Omega$ (which is actually implied by the other four equations). 

We start by showing that the Einstein equations imply the first order system of PDE~\eqref{r_u}$-$\eqref{kappa_at_u}.
More precisely, this system consists of ten equations for the seven unknowns $r$, $\nu$, $\lambda$, $\varpi$, $\theta$, $\zeta$ and $\kappa$ (with $\nu=\partial_ur$, $\lambda=\partial_vr$,
$\theta=r\partial_v\phi$ and $\zeta=r\partial_u\phi$): two equations for the first partial derivatives of~$r$, two equations which correspond to the wave equation for~$r$,
two equations for the first partial derivatives of the renormalized mass $\varpi$, two equations which correspond to the wave equation for $\phi$, an equation for
a derivative of $\kappa$, and an algebraic restriction equation $\lambda=\kappa(\muu)$ (which can be thought of as the definition of $\kappa$).
The introduction of $\kappa$ allows us to avoid unpleasant denominators which a priori could vanish.
Since there are more equations than unknowns, we solve seven of these equations and treat the remaining three as restrictions, which we later check to be
preserved by the flow.

Our first task is to prove existence of solutions for the first order system (Theorem~\ref{existence}). 
The reason why we work with this system and not with the Einstein equations is that we need to prove existence of solutions defined in rectangles $[0,U]\times[0,V]$
for any given $V$ (and suitably small $U$), or any given $U$ (and suitably small $V$); it is not enough to prove such a result for both $U$ and $V$ small.
By working with the first order system we are in fact keeping track of more functions and their derivatives, providing a finer control of $r$, $\Omega$ and $\phi$.
The proof follows a standard Banach Fixed Point Theorem argument. The construction of the contracting operator and the corresponding complete metric space is somewhat subtle, and uses the specific form of our system.

After proving an appropriate uniqueness result, we show that any solution can be extended uniquely to a maximal past set ${\cal P}$ (Theorem~\ref{maximal}).
We then proceed to establish a breakdown criterion for the solution: along all sequences in ${\cal P}$ converging to the boundary of ${\cal P}$,
the radial coordinate $r$ goes to zero and the renormalized mass $\varpi$ goes to $+\infty$ (Theorem~\ref{breakdown}).

We finish Part~1 by showing that, under a stronger regularity condition on the initial data, a solution of the first order system corresponds, in fact, to a solution of the Einstein equations
(Proposition~\ref{einstein-2}). Furthermore, additional regularity for the initial data implies additional regularity for the solutions (Proposition~\ref{regularity} and Remark~\ref{rmk-reg}).

\section{Derivation of a first order system from the Einstein equations}
We consider the Einstein-Maxwell-real massless scalar field equations in the presence of a cosmological constant $\Lambda$ (in units for which $c=4\pi G=\varepsilon_0=1$):
\begin{align*}
& R_{\mu\nu} - \frac12 R g_{\mu\nu} + \Lambda g_{\mu\nu} = 2 T_{\mu\nu}; \\
& dF = d^{*}F = 0; \\
& \Box \phi = 0; \\
& T_{\mu\nu} = \partial_\mu \phi \, \partial_\nu \phi - \frac12 \partial_\alpha \phi \, \partial^\alpha \phi \, g_{\mu\nu} + F_{\mu\alpha} F_{\nu}^{\,\,\alpha} - \frac14 F_{\alpha\beta} F^{\alpha \beta} g_{\mu\nu}.
\end{align*}
These form a system of partial differential equations for the components of the spacetime metric $g$, the Faraday electromagnetic $2$-form $F$ and the real massless scalar field $\phi$; here $R_{\mu\nu}$ are the components of the Ricci tensor, $R$ is the scalar curvature, $^*$ is the Hodge star operator and $\Box$ is the d'Alembertian (all depending on $g$).

In the spherically symmetric case, we can write the metric in double null coordinates $(u,v)$ as 
\[
g=-\Omega^2(u,v)\,dudv+r^2(u,v)\,\sigma_{\mathbb{S}^2}, 
\]
where $\sigma_{\mathbb{S}^2}$ is the round metric on the $2$-sphere $\mathbb{S}^2$. In this case, the Maxwell equations decouple from the system, since they can be immediately solved to yield
\[
F = - \frac{e \, \Omega^2(u,v)}{2 \, r^2(u,v)} \, du \wedge dv.
\]
Here $e$ is a real constant, corresponding to a total electric charge $4 \pi e$, and we have assumed zero magnetic charge without loss of generality.

The remaining equations can then be written as follows (this is a straightforward modification of the equations in \cite{DafermosBlack} and \cite{HolzegelSelf}, see also \cite{ChristodoulouTwo}):
a wave equation for $r$,
\begin{equation}\label{wave_r} 
\partial_u\partial_vr=-\frac{\Omega^2}{4r} - \frac{\partial_ur\,\partial_vr}{r} + \frac{\Omega^2e^2}{4r^3} + \frac{\Omega^2 \Lambda r}{4},
\end{equation}
a wave equation for $\phi$,
\begin{equation}\label{wave_phi} 
\partial_u\partial_v\phi=-\,\frac{\partial_ur\,\partial_v\phi+\partial_vr\,\partial_u\phi}{r},
\end{equation}
the Raychaudhuri equation in the $u$ direction,
\begin{equation}\label{r_uu} 
\partial_u\left(\frac{\partial_ur}{\Omega^2}\right)=-r\frac{(\partial_u\phi)^2}{\Omega^2},
\end{equation}
the Raychaudhuri equation in the $v$ direction,
\begin{equation}\label{r_vv} 
\partial_v\left(\frac{\partial_vr}{\Omega^2}\right)=-r\frac{(\partial_v\phi)^2}{\Omega^2},
\end{equation}
and a wave equation for $\ln\Omega$,
 \begin{equation}\label{wave_Omega} 
\partial_v\partial_u\ln\Omega=-\partial_u\phi\,\partial_v\phi-\,\frac{\Omega^2e^2}{2r^4}+\frac{\Omega^2}{4r^2}+\frac{\partial_ur\,\partial_vr}{r^2}.
\end{equation}

To write the Einstein equations as a first order system of PDE we define the following quantities:
\begin{equation}\label{nu_0}
\nu:=\partial_u r,
\end{equation}
\begin{equation}\label{lambda_0}
\lambda:=\partial_v r,
\end{equation}
\begin{equation}\label{bar_rafaeli} 
\omega:=\frac{e^2}{2r}+\frac{r}{2}-\frac{\Lambda}{6}r^3+\frac{2r}{\Omega^2}\nu\lambda,
\end{equation}
\begin{equation}\label{mu_0} 
\mu:=\truemu,
\end{equation}
\begin{equation}\label{theta} 
\theta:=r\partial_v\phi,
\end{equation}
\begin{equation}\label{zeta} 
\zeta:=r\partial_u\phi
\end{equation}
and
\begin{equation}\label{kappa_0} 
 \kappa:=\frac{\lambda}{1-\mu}.
\end{equation}
Notice that we may rewrite~\eqref{bar_rafaeli} as
\begin{equation}\label{omega_sq}
\Omega^2=-\,\frac{4\nu\lambda}{1-\mu}=-4\nu\kappa.
\end{equation}
It is easy to see that
\[
1-\mu = g(\nabla r,\nabla r). 
\]
Therefore $\varpi$, like $r$, is a geometric quantity: it is called the renormalized Hawking mass.

\begin{Prop}\label{Erica}
The Einstein equations~\eqref{wave_r}$-$\eqref{r_vv} imply the following first order differential equations:
\begin{align*}
& \partial_u\lambda=\partial_v\nu=-2\nu\kappa\frac{1}{r^2}\left(\vlinha\right), \\
& \partial_u\omega=\frac 12\left(\muu\right)\left(\frac\zeta\nu\right)^2\nu, \\
& \partial_v\omega=\frac 12\frac{\theta^2}{\kappa}, \\
& \partial_u\theta=-\,\frac{\zeta\lambda}{r}, \\
& \partial_v\zeta=-\,\frac{\theta\nu}{r}, \\
& \partial_u\kappa=\kappa\nu\frac 1r\left(\frac{\zeta}{\nu}\right)^2.
\end{align*}
\end{Prop}
\begin{proof}
The wave equation~\eqref{wave_r} written in terms of $\lambda$ in~\eqref{lambda_0} is
\begin{equation}\label{l_u}
\partial_u\lambda=-2\nu\kappa\frac{1}{r^2}\left(\vlinha\right),
\end{equation}
and written in terms of $\nu$ in~\eqref{nu_0} is
\begin{equation}\label{n_v}
\partial_v\nu=-2\nu\kappa\frac{1}{r^2}\left(\vlinha\right),
\end{equation}
where we have used~\eqref{omega_sq}.

The wave equation~\eqref{wave_phi} written in terms of $\theta$ in~\eqref{theta} is
$$
 \partial_u\theta=-\partial_v r\partial_u\phi=-\,\frac{\zeta\lambda}{r}.
$$
For the last equality we have used~\eqref{lambda_0} and~\eqref{zeta}.
The wave equation~\eqref{wave_phi} written in terms of $\zeta$ in~\eqref{zeta} is
$$
 \partial_v\zeta=-\partial_u r\partial_v\phi=-\,\frac{\theta\nu}{r}.
$$
For the last equality we have used~\eqref{nu_0} and~\eqref{theta}.

Using \eqref{nu_0}, \eqref{zeta} and \eqref{omega_sq}, we may rewrite
the Raychaudhuri equation~\eqref{r_uu} as
\begin{equation}\label{ray_u} 
\partial_u\left(\frac{1-\mu}{\lambda}\right)=-\frac{\zeta^2}{r\nu\lambda}(1-\mu),
\end{equation}
and using \eqref{lambda_0}, \eqref{theta} and \eqref{omega_sq}, we may rewrite the Raychaudhuri equation~\eqref{r_vv} as
\begin{equation}\label{ray_v} 
\partial_v\left(\frac{1-\mu}{\nu}\right)=-\frac{\theta^2}{r\nu\lambda}(1-\mu).
\end{equation}
Expanding the left-hand sides of~\eqref{ray_u} and~\eqref{ray_v}, taking into account the definition of $\mu$ in~\eqref{mu_0},
and the values of the derivatives~\eqref{nu_0} and \eqref{l_u}, we obtain
$$
-\,\frac{2\partial_u\omega}{r\lambda}=-\frac{\zeta^2}{r\nu\lambda}(1-\mu)\ \Leftrightarrow\ \partial_u\omega=\frac 12(1-\mu)\frac{\zeta^2}{\nu},
$$
and taking into account the values of the derivatives~\eqref{lambda_0} and \eqref{n_v},
we obtain
$$
-\,\frac{2\partial_v\omega}{r\nu}=-\frac{\theta^2}{r\nu\lambda}(1-\mu)\ \Leftrightarrow\ \partial_v\omega=\frac 12(1-\mu)\frac{\theta^2}{\lambda}.
$$

Finally, from~\eqref{kappa_0} and \eqref{ray_u}, it follows that
\begin{eqnarray*}
 \partial_u\kappa&=&\partial_u\left(\frac{\lambda}{1-\mu}\right)\ =\ -\left(\frac{\lambda}{1-\mu}\right)^2\partial_u\left(\frac{1-\mu}{\lambda}\right)\
 =\ \kappa\frac 1r\frac{\zeta^2}{\nu}.
\end{eqnarray*}
\end{proof}

\section{Existence of solution of the first order system}

As mentioned in the introduction, we use the differential equations in Proposition~\ref{Erica} to set up the following closed first order system for the seven unknowns $r$, $\nu$, $\lambda$, $\varpi$, $\theta$, $\zeta$ and $\kappa$:
\begin{eqnarray} 
 \partial_ur&=&\nu\label{r_u},\\
 \partial_vr&=&\lambda\label{r_v},\\
 \partial_u\lambda&=&-2\nu\kappa\frac{1}{r^2}\left(\vlinha\right)\label{lambda_u},\\
 \partial_v\nu&=&-2\nu\kappa\frac{1}{r^2}\left(\vlinha\right),\label{nu_v}\\
 \partial_u\omega&=&\frac 12\left(\muu\right)\left(\frac\zeta\nu\right)^2\nu,\label{omega_u}\\
 \partial_v\omega&=&\frac 12\frac{\theta^2}{\kappa},\label{omega_v}\\
 \partial_u\theta&=&-\,\frac{\zeta\lambda}{r},\label{theta_u}\\
 \partial_v\zeta&=&-\,\frac{\theta\nu}{r},\label{zeta_v}\\
 \partial_u\kappa&=&\kappa\nu\frac 1r\left(\frac{\zeta}{\nu}\right)^2,\label{kappa_u}
\end{eqnarray}
subject to the algebraic constraint
\begin{equation}\label{kappa_at_u} 
\lambda=\kappa\left(\muu\right).
\end{equation}

Notice that this system is overdetermined, as there are two equations for $r$ and two equations for $\omega$, besides~\eqref{kappa_at_u}. We choose to regard equations~\eqref{r_u}, \eqref{omega_v} and~\eqref{kappa_at_u} as constraints on the initial data; we will then show that they are preserved by the evolution. Therefore, we solve equations~\eqref{r_v}$-$\eqref{omega_u} and \eqref{theta_u}$-$\eqref{kappa_u} only, with initial data satisfying appropriate restrictions. This will suffice to retrieve a solution of the full system \eqref{r_u}$-$\eqref{kappa_at_u}, which, under suitable regularity assumptions, yields a solution of the Einstein equations~\eqref{wave_r}$-$\eqref{wave_Omega} (see Section~\ref{Miranda}).

We take the domain of the solution to be a subset of $[0,U]\times[0,+\infty[$, with initial conditions prescribed as follows:
$$
({\rm I}_u)\qquad\left\{
\begin{array}{lcl}
 r(u,0)&=&r_0(u),\\
 \nu(u,0)&=&\nu_0(u),\\
 \zeta(u,0)&=&\zeta_0(u),
\end{array}
\right.\qquad{\rm for}\ u\in[0,U],
$$
$$
({\rm I}_v)\qquad\left\{
\begin{array}{lcl}
 \lambda(0,v)&=&\lambda_0(v),\\
 \omega(0,v)&=&\omega_0(v),\\
 \theta(0,v)&=&\theta_0(v),\\
 \kappa(0,v)&=&\kappa_0(v),
\end{array}
\right.\qquad{\rm for}\ v\in[0,\infty[.
$$
Let $$\tilde{r}_0(v)=r_0(0)+\int_0^v\lambda_0(v')\,dv',$$ for $v\in[0,+\infty[$.
We assume the regularity conditions:
\begin{align*}
\text{(h1)} \qquad &\text{the functions } \nu_0, \zeta_0, \lambda_0, \theta_0 \text{ and } \kappa_0 \text{ are continuous, and} \\
&\text{the functions } r_0 \text{ and } \omega_0 \text{ are continuously differentiable.}
\end{align*}
We assume the sign conditions:
$$
\hspace{-4cm}
\text{(h2)} \qquad 
\left\{
\begin{array}{ll}
r_0(u)>0&{\rm for}\ u\in[0,U],\\
\tilde{r}_0(v)>0&{\rm for}\ v\in[0,+\infty[,\\
\nu_0(u)<0&{\rm for}\ u\in[0,U],\\
\kappa_0(v)>0&{\rm for}\ v\in[0,\infty[.
\end{array}\right.
$$
Note that the condition on $\tilde{r}_0$ is actually a restriction on $\lambda_0$.
We assume the three compatibility conditions:
\begin{align} 
& r_0'=\nu_0,\label{initial_rp}\\
\text{(h3)} \qquad & \omega_0'=\frac 12\frac{\theta_0^2}{\kappa_0},\label{initial_v}\\
& \lambda_0=\kappa_0\left(\muuz\right),\label{kappa_at_zero} \hspace{3cm}
\end{align}
so that the initial data satisfy~(\ref{r_u}), (\ref{omega_v}) and (\ref{kappa_at_u}).

One way to guarantee that (h2)$-$(h3) are satisfied is the following. First, choose $r_0(0)>0$ and $\nu_0<0$, and use (\ref{initial_rp}) to compute $r_0$.
Next choose $\omega_0(0)$, $\theta_0$ and $\kappa_0>0$, and use (\ref{initial_v}) to compute $\omega_0$.
Finally, solve (\ref{kappa_at_zero}) as an ordinary differential equation for $\tilde{r}_0$, and then obtain $\lambda_0$ as its derivative.
In the end, one needs to check that the choices made lead to $r_0>0$ and $\tilde{r}_0>0$.

In the case where we take the data on the event horizon of the Reissner-Nordstr\"{o}m solution
with mass $M$ as the initial data on the $v$ axis, we may choose
\begin{equation}\label{RN1} 
\left\{
\begin{array}{lcl}
 r(u,0)&=&r_+-u,\\
 \nu(u,0)&=&-1,\\
 \zeta(u,0)&=&\zeta_0(u),
\end{array}
\right.\qquad{\rm for}\ u\in[0,U],
\end{equation}
and
\begin{equation}\label{RN2} 
\left\{
\begin{array}{lcl}
 \lambda(0,v)&=&0,\\
 \omega(0,v)&=&M,\\
 \theta(0,v)&=&0,\\
 \kappa(0,v)&=&1,
\end{array}
\right.\qquad{\rm for}\ v\in[0,\infty[.
\end{equation}
Here $\zeta_0(0)$ should be zero.
It is easily seen that (h2)$-$(h3) are satisfied in this situation if $U<r_+$.

The choices of $\lambda_0$ and $\nu_0$ correspond to fixing the coordinate system.
If $\lambda_0\neq 0$, the functions $\nu_0$ and $\lambda_0$ determine the coordinates $u$ and $v$, respectively, as functions of $r$, along the initial null rays.
If $\lambda_0=0$, it is the choice of $\kappa_0$ that fixes the coordinate $v$. 
Notice that the latter situation occurs for the Reissner-Nordstr\"{o}m initial data, which we will address in Parts~2 and~3. In fact, the additional unknown $\kappa$ was introduced precisely to overcome
the difficulties with the case $\lambda_0=0$, as $\frac{\lambda}{1-\mu}$ would otherwise be indeterminate at the event horizon.

\vspace{4mm}

Recall that ${\cal P}\subset [0,U]\times[0,+\infty[$ is a past set if
$J^-(u,v):=[0,u]\times[0,v]$ is a subset of ${\cal P}$ for all $(u,v)\in{\cal P}$.
\begin{Def}[Solution of the characteristic initial value problem]\label{solution} 
A solution
of the characteristic initial value problem is a set of continuous functions $r$, $\lambda$, $\nu$,
$\omega$, $\theta$, $\zeta$ and $\kappa$ defined on a past set ${\cal P}$ containing $(0,0)$, with $r$, $\nu$ and $\kappa$ nonzero, such that\/
 {\rm (\ref{r_u})$-$(\ref{kappa_at_u})} are satisfied, with
all the partial derivatives occurring in\/
 {\rm (\ref{r_u})$-$(\ref{kappa_u})} continuous. Furthermore the
initial conditions {\rm (I$_u$)} and {\rm (I$_v$)} are satisfied on the intersection of the past set with $[0,U]\times\{0\}\cup\{0\}\times[0,+\infty[$.
\end{Def}

Let us make an observation concerning the gluing of solutions defined on rectangles, which will be used throughout this work. Suppose that we have a solution of~\eqref{r_u}$-$\eqref{kappa_at_u} in a closed rectangle ${\cal R}_1$ and another solution in a closed rectangle ${\cal R}_2$, with ${\cal R}_1$ and ${\cal R}_2$ disjoint except for a common edge. Suppose also that all functions coincide on ${\cal R}_1\cap {\cal R}_2$. Then the obvious extension determines a solution of~\eqref{r_u}$-$\eqref{kappa_at_u} on ${\cal R}_1\cup {\cal R}_2$, since the extended functions are clearly continuous and the equations imply the continuity of the relevant partial derivatives across the common edge.

\begin{Thm}[Existence of solution in some rectangles]\label{existence} 
Given $0<V<\infty$, there exists $0<\tilde{U}\leq U$ such that the characteristic initial value problem with
initial data\/~{\rm (I$_u$)} and\/~{\rm (I$_v$)} satisfying\/ {\rm (h1)}$-${\rm (h3)}
has a unique solution in the rectangle\/ $[0,\tilde{U}]\times[0,V]$.
Also, there exists $0<\tilde{V}<\infty$ such that the characteristic initial value problem with
initial data satisfying\/ {\rm (h1)}$-${\rm (h3)}
has a unique solution in the rectangle\/ $[0,U]\times[0,\tilde{V}]$.
The values of $\tilde{U}$ and $\tilde{V}$ depend only on the initial data.
\end{Thm}
\begin{proof}
We define, by induction, the sequence of functions
\begin{eqnarray*}
\lambda_{n+1}(u,v)&=&\lambda_0(v)-\int_0^u\left[2\nu\kappa\frac{1}{r^2}\left(\vlinha\right)\right]_n(u',v)\,du',\\
\omega_{n+1}(u,v)&=&\omega_0(v)e^{-\int_0^u\bigl[\frac{\zeta^2}{r\nu}\bigr]_n(u',v)\,du'}\\ &&+
\int_0^ue^{-\int_s^u\bigl[\frac{\zeta^2}{r\nu}\bigr]_n(u',v)\,du'}
\left[\frac{1}{2}\left(1+\frac{e^2}{r^2}-\frac{\Lambda}{3}r^2\right)\frac{\zeta^2}{\nu}\right]_n(s,v)\,ds,\\
\theta_{n+1}(u,v)&=&\theta_0(v)-\int_0^u\left[\frac{\zeta\lambda}{r}\right]_n(u',v)\,du',\\
\kappa_{n+1}(u,v)&=&\kappa_0(v)e^{\int_0^u\bigl[\frac{\zeta^2}{r\nu}\bigr]_n(u',v)\,du'},\\
r_{n+1}(u,v)&=&r_0(u)+\int_0^v\lambda_{n+1}(u,v')\,dv',\\
\nu_{n+1}(u,v)&=&\nu_0(u)e^{-\int_0^v\bigl[2\kappa\frac{1}{r^2}\left(\vlinha\right)\bigr]_{n+1}(u,v')\,dv'},\\
\zeta_{n+1}(u,v)&=&\zeta_0(u)-\int_0^v\left[\frac{\theta\nu}{r}\right]_{n+1}(u,v')\,dv',
\end{eqnarray*}
in this order. Note that the integrals in $u$ involve functions computed in the previous iteration $n$, whereas the integrals in $v$
use functions already computed in the current iteration $n+1$.

We start with $\lambda_1(u,v)=\lambda_0(v)$, $\omega_1(u,v)=\omega_0(v)$, $\theta_1(u,v)=\theta_0(v)$ and $\kappa_1(u,v)=\kappa_0(v)$.
We use the last three equations in the iteration scheme to compute $r_1(u,v)$, $\nu_1(u,v)$ and $\zeta_1(u,v)$ (in this order).
Let $V$ be arbitrary but fixed in $]0,+\infty[$. Let $0<\tilde{U}\leq U$. Consider the (metric) space ${\cal S}$ of functions
$(r,\nu,\lambda,\omega,\theta,\zeta,\kappa)$ defined as
$$I_{[{\underline{r}},{\overline{r}}]}(r_0) \times B_{R_\nu}(\nu_0) \times B_{R_\lambda}(\lambda_0)\times B_{R_\omega}(\omega_0) \times
B_{R_\theta}(\theta_0)\times B_{R_\zeta}(\zeta_0)\times B_{R_\kappa}(\kappa_0)$$
equipped with the product of $L^\infty$ norms.
Here
$$
 I_{[{\underline{r}},{\overline{r}}]}(r_0)=\{r\in C^0([0,\tilde{U}]\times[0,V]):{\underline{r}}\leq r(u,v)-r_0(u) \leq {\overline{r}}\},
$$
the ball $B_{R_\nu}(\nu_0)$ is defined by
$$
B_{R_\nu}(\nu_0)=\{\nu\in C^0([0,\tilde{U}]\times[0,V]):\|\nu(u,v)-\nu_0(v)\|_{L^\infty([0,\tilde{U}]\times[0,V])}\leq R_\nu\},
$$
and the other balls are defined similarly.

We wish to choose $\tilde{U}$, the radii of the six balls in the definition of ${\cal S}$,
$\underline{r}$ and $\overline{r}$
so that the sequence constructed by the iteration scheme
lies in ${\cal S}$.
Assume that
$$
(\lambda_n,\omega_n,\theta_n,\kappa_n)\in B_{R_\lambda}(\lambda_0)\times B_{R_\omega}(\omega_0) \times
B_{R_\theta}(\theta_0)\times B_{R_\kappa}(\kappa_0),
$$
where the values of $R_\omega$, $R_\theta$ and $R_\kappa$ are fixed freely and $R_\lambda$ is specified below.
The function $\tilde{r}_0(v)$ has a positive minimum in $[0,V]$ equal to $\tilde{r}_{\min}$,
and a positive maximum equal to $\tilde{r}_{\max}$.
We have
\begin{eqnarray*}
r_{n}(u,v)&=&r_0(u)+\int_0^v\lambda_{n}(u,v')\,dv'\\
&=&r_0(u)-r_0(0)+\tilde{r}_0(v)+\int_0^v(\lambda_{n}(u,v')-\lambda_0(v'))\,dv'\\
&\geq&r_0(u)-r_0(0)+\tilde{r}_{\min}-R_\lambda V.
\end{eqnarray*}
Since $r_0$ is continuous, we can choose $\tilde{U}$ and $R_\lambda$ sufficiently small so that
$r_{n}(u,v)\geq\frac{\tilde{r}_{\min}}{2}$ for $(u,v)\in[0,\tilde{U}]\times[0,V]$.
Analogously,
\begin{eqnarray*}
r_{n}(u,v)
&\leq&r_0(u)-r_0(0)+\tilde{r}_{\max}+R_\lambda V.
\end{eqnarray*}
We choose ${\underline{r}}=-r_0(0)+\tilde{r}_{\min}-R_\lambda V$ and we choose
${\overline{r}}=-r_0(0)+\tilde{r}_{\max}+R_\lambda V$. This guarantees that $r_{n}$ lies in $I_{[{\underline{r}},{\overline{r}}]}(r_0)$.
We now choose $R_\nu$ sufficiently large
(depending on $V$, $R_\omega$, $R_\kappa$ and the upper and lower bounds on $r$)
so that the iteration scheme produces functions $\nu_{n}$ that lie in
$B_{R_\nu}(\nu_0)$. Moreover, the function $\nu_n$ is bounded away from zero by a constant not depending on $n$.
At this point we may choose $R_\zeta$ sufficiently large so that the iteration scheme produces functions $\zeta_{n}$ that lie in
$B_{R_\zeta}(\zeta_0)$.
We further reduce $\tilde{U}$, if necessary, so that the iteration scheme produces functions
$\lambda_{n+1}$,
$\omega_{n+1}$,
$\theta_{n+1}$ and
$\kappa_{n+1}$ that lie in $B_{R_\lambda}(\lambda_0)$, $B_{R_\omega}(\omega_0)$, $B_{R_\theta}(\theta_0)$ and $B_{R_\kappa}(\kappa_0)$,
respectively. This shows that the operator ${\cal T}$ that sends $(r,\nu,\lambda,\omega,\theta,\zeta,\kappa)_n$ to
$(r,\nu,\lambda,\omega,\theta,\zeta,\kappa)_{n+1}$ is an operator from ${\cal S}$ to ${\cal S}$.

Since the integrals in $u$ involve functions computed in the iteration $n$, by using the mean value theorem, it is clear that the norms of
$\lambda_{n+1}-\lambda_n$,
$\omega_{n+1}-\omega_n$,
$\theta_{n+1}-\theta_n$ and
$\kappa_{n+1}-\kappa_n$ in $L^\infty([0,\tilde{U}]\times[0,V])$ are bounded by $C\tilde{U}d_n$, for some positive constant $C$, with
$$d_n:=\|(r,\nu,\lambda,\omega,\theta,\zeta,\kappa)_{n}-(r,\nu,\lambda,\omega,\theta,\zeta,\kappa)_{n-1}\|_{{\cal S}}.$$
We now estimate the difference of consecutive iterates of $r$, $\nu$ and $\zeta$, in this order.
We have
$$\|r_{n+1}-r_n\|\leq V\|\lambda_{n+1}-\lambda_n\|\leq VC\tilde{U}d_n.$$
Then, there exists another positive constant $C$ such that
$$\|\nu_{n+1}-\nu_n\|\leq C(\|\omega_{n+1}-\omega_n\|+\|\kappa_{n+1}-\kappa_n\|+\|r_{n+1}-r_n\|)
\leq C\tilde{U}d_n.$$
Finally, for some other constant $C$, we have
$$\|\zeta_{n+1}-\zeta_n\|\leq C(\|\theta_{n+1}-\theta_n\|+\|r_{n+1}-r_n\|+\|\nu_{n+1}-\nu_n\|)
\leq C\tilde{U}d_n.$$
In conclusion, $d_{n+1}\leq C\tilde{U}d_n$.
By further reducing $\tilde{U}$ if necessary, ${\cal T}$ is a contraction. By the Banach fixed point theorem, ${\cal T}$
has a unique fixed point in ${\cal S}$.

\vspace{4mm}

We will check below that this fixed point of ${\cal T}$ is the solution to our first order system in a rectangle whose projection
on the $v$ axis is fixed, the interval $[0,V]$, and whose projection on the $u$ axis is small, $[0,\tilde{U}]$.
A similar argument allows us to construct a solution in a rectangle whose projection
on the $u$ axis is fixed, the interval $[0,U]$, and whose projection on the $v$ axis is small, $[0,\tilde{V}]$.
The argument runs as follows. Define, by induction, the sequence of functions
\begin{eqnarray}
r_{n+1}(u,v)&=&r_0(u)+\int_0^v\lambda_{n}(u,v')\,dv',\nonumber\\
\nu_{n+1}(u,v)&=&\nu_0(u)e^{-\int_0^v\bigl[2\kappa\frac{1}{r^2}\left(\vlinha\right)\bigr]_{n}(u,v')\,dv'},\nonumber\\
\zeta_{n+1}(u,v)&=&\zeta_0(u)-\int_0^v\left[\frac{\theta\nu}{r}\right]_{n}(u,v')\,dv',\nonumber\\
\omega_{n+1}(u,v)&=&\omega_0(v)e^{-\int_0^u\bigl[\frac{\zeta^2}{r\nu}\bigr]_{n+1}(u',v)\,du'}\label{omega_mod}\\ &&+
\int_0^ue^{-\int_s^u\bigl[\frac{\zeta^2}{r\nu}\bigr](u',v)\,du'}
\left[\frac{1}{2}\left(1+\frac{e^2}{r^2}-\frac{\Lambda}{3}r^2\right)\frac{\zeta^2}{\nu}\right]_{n+1}(s,v)\,ds,\nonumber\\
\kappa_{n+1}(u,v)&=&\kappa_0(v)e^{\int_0^u\bigl[\frac{\zeta^2}{r\nu}\bigr]_{n+1}(u',v)\,du'},\label{kappa_mod}\\
\lambda_{n+1}(u,v)&=&\lambda_0(v)-\int_0^u\left[2\nu\kappa\frac{1}{r^2}\left(\vlinha\right)\right]_{n+1}(u',v)\,du',\label{lambda_mod}\\
\theta_{n+1}(u,v)&=&\theta_0(v)-\int_0^u\left[\frac{\zeta\lambda}{r}\right]_{n+1}(u',v)\,du',\label{theta_mod}
\end{eqnarray}
in this order. Notice that it is the integrals in $v$ that now involve functions computed in the iteration $n$, while the integrals in $u$
use the functions in the current iteration $n+1$.

Start with $r_1(u,v)=r_0(u)$, $\nu_1(u,v)=\nu_0(u)$ and $\zeta_1(u,v)=\zeta_0(u)$. Use the last four equations in
the iteration scheme to compute $\omega_1(u,v)$, $\kappa_1(u,v)$, $\lambda_1(u,v)$, $\theta_1(u,v)$ (in this order).
Let $0<\tilde{V}<\infty$. Consider the same space ${\cal S}$ of functions 
as above, but now with domain $[0,U]\times[0,\tilde{V}]$.

We wish to choose $\tilde{V}$, the radii of the six balls in the definition of ${\cal S}$,
$\underline{r}$ and $\overline{r}$
so that the sequence constructed by the iteration scheme
lies in ${\cal S}$.
Assume that
$$
(r_n,\nu_n,\zeta_n)\in
I_{[{\underline{r}},{\overline{r}}]}(r_0) \times B_{R_\nu}(\nu_0) \times B_{R_\zeta}(\zeta_0),$$
where the values of $R_\nu$ and $R_\zeta$ are fixed freely and ${\underline{r}}$ and ${\overline{r}}$ are specified below.

Let $r_{\min}$ be the minimum of $r_0(u)$ on $[0,U]$ (that is $r_0(U)$).
Fix $-r_{\min}<\underline{r}<0<\overline{r}$, so that functions $r$ in $I_{[\underline{r},\overline{r}]}(r_0)$ are bounded below by a
strictly positive constant.
Now choose $R_\omega$ and $R_\kappa$ sufficiently large so that the iteration scheme produces functions $\omega_{n}$
and $\kappa_{n}$ that lie in
$B_{R_\omega}(\omega_0)$ and $B_{R_\kappa}(\kappa_0)$.
At this point choose $R_\lambda$ sufficiently large so that the iteration scheme produces functions $\lambda_{n}$ that lie in
$B_{R_\lambda}(\lambda_0)$. Having done so, choose $R_\theta$ sufficiently large so that the iteration scheme produces functions $\theta_{n}$ that lie in
$B_{R_\theta}(\theta_0)$. Choose $\tilde{V}$ sufficiently small so that the iteration scheme produces functions $r_{n+1}$, $\nu_{n+1}$
and $\zeta_{n+1}$ that lie in $I_{[\underline{r},\overline{r}]}(r_0)$, $B_{R_\nu}(\nu_0)$ and $B_{R_\zeta}(\zeta_0)$, respectively.
By further reducing $\tilde{V}$, if necessary,  the operator that sends $(r,\nu,\lambda,\omega,\theta,\zeta,\kappa)_n$ to
$(r,\nu,\lambda,\omega,\theta,\zeta,\kappa)_{n+1}$ has a unique fixed point in ${\cal S}$.

\vspace{4mm}

Obviously, in both cases, functions with domain $[0,\tilde{U}]\times[0,V]$ and functions with domain $[0,U]\times[0,\tilde{V}]$,
the fixed points satisfy
\begin{eqnarray} 
\kappa(u,v)&=&\kappa_0(v)e^{\int_0^u\bigl(\frac{\zeta^2}{r\nu}\bigr)(u',v)\,du'},\label{kappa_final}\\
\nu(u,v)&=&\nu_0(u)e^{-\int_0^v\left(2\kappa\frac{1}{r^2}\left(\vlinha\right)\right)(u,v')\,dv'},\label{nu_final}\\
\lambda(u,v)&=&\lambda_0(v)-\int_0^u\left(2\nu\kappa\frac{1}{r^2}\left(\vlinha\right)\right)(u',v)\,du',\label{lambda_final}\\
\theta(u,v)&=&\theta_0(v)-\int_0^u\left(\frac{\zeta\lambda}{r}\right)(u',v)\,du',\label{theta_final}\\
\zeta(u,v)&=&\zeta_0(u)-\int_0^v\left(\frac{\theta\nu}{r}\right)(u,v')\,dv',\label{zeta_final}\\
\omega(u,v)&=&\omega_0(v)e^{-\int_0^u\bigl(\frac{\zeta^2}{r\nu}\bigr)(u',v)\,du'}\nonumber\\ &&+
\int_0^ue^{-\int_s^u\frac{\zeta^2}{r\nu}(u',v)\,du'}\left(\frac{1}{2}\left(1+\frac{e^2}{r^2}
-\frac{\Lambda}{3}r^2\right)\frac{\zeta^2}{\nu}\right)(s,v)\,ds,\label{omega_final}\\
r(u,v)&=&r_0(u)+\int_0^v\lambda(u,v')\,dv'.\label{r_final}
\end{eqnarray}

As a consequence, (\ref{r_v})$-$(\ref{omega_u}) and (\ref{theta_u})$-$(\ref{kappa_u}) hold in the sense of Definition~\ref{solution}.

\vspace{4mm}

We now show that $\partial_u r=\nu$, which we imposed at $v=0$ in (\ref{initial_rp}), propagates to all values of $v$.
Clearly, from (\ref{lambda_final}) and (\ref{r_final}), $\partial_u\partial_v r=\partial_u\lambda$ is continuous.
 Equation (\ref{r_final})
 implies that $\partial_u r$ is continuous, and so $r$ is $C^1$.
Therefore the derivative $\partial_v\partial_u r$ exists and $\partial_v\partial_u r=\partial_u\partial_v r$. So, by \eqref{lambda_u} and \eqref{nu_v},
$$
\partial_v(\partial_u r-\nu)=\partial_u\lambda-\partial_v\nu=0.
$$
This establishes (\ref{r_u}) in the whole rectangle.

Next we check that (\ref{kappa_at_u}), which we imposed at $u=0$ in (\ref{kappa_at_zero}), propagates to all values of $u$. Define $\mu$ using~\eqref{mu_0}.
Equation (\ref{kappa_at_u}) is established in the whole domain by checking that
$$
\partial_u(\lambda-\kappa(1-\mu))=0,
$$
which follows from \eqref{r_u}, \eqref{lambda_u}, \eqref{omega_u} and \eqref{kappa_u}.

Finally we show that (\ref{omega_v}), which we imposed at $u=0$ in (\ref{initial_v}), propagates to all values of $u$.
From (\ref{omega_final}),
$\partial_v\omega$ is continuous as $\nu$ and $\zeta$ have continuous derivative with respect to $v$ from \eqref{nu_final} and \eqref{zeta_final}.
Using (\ref{omega_u}), this in turn implies that
$
\partial_v\partial_u\omega
$
is continuous. Therefore the derivative $\partial_u\partial_v \omega$ exists and  $\partial_u\partial_v \omega=\partial_v\partial_u \omega$. So,
a straightforward computation yields
\begin{eqnarray*}
\partial_u\left(\partial_v\omega-\frac 12\frac{\theta^2}{\kappa}\right)&=&\partial_v\partial_u\omega-\partial_u\frac 12\frac{\theta^2}{\kappa}\\
&=&-\,\frac 1r\left(\frac{\zeta}{\nu}\right)^2\nu\left(\partial_v\omega-\frac 12\frac{\theta^2}{\kappa}\right).
\end{eqnarray*}
From (\ref{initial_v}), we conclude that (\ref{omega_v}) holds for all $u$.

Notice that we have only established uniqueness of the solution in ${\cal S}$.
The uniqueness assertion in the statement of the theorem is a special case of Proposition~\ref{uniqueness} below.
\end{proof}

\begin{Prop}[Uniqueness]\label{uniqueness}
Two solutions of the characteristic initial value problem with the same initial conditions coincide on the intersection of their domains.
\end{Prop}
\begin{proof}
The intersection of the domains is a past set.
Let $(u,v)$ be a point on this intersection. Then all functions are bounded on the rectangle $[0,u]\times[0,v]$.
Furthermore, there exists a positive lower bound for $r$ and $|\nu|$ for both solutions.
Let $u_0\in[0,u[$.
Theorem~\ref{existence} guarantees the existence of a $\tilde{U}$, depending only on these uniform bounds, such
that the characteristic initial value problem with initial data on $\{u_0\}\times[0,v]$ and on $[u_0,u_0+\tilde{U}]\times\{0\}$
has a unique solution on $[u_0,u_0+\tilde{U}]\times[0,v]$ in the set ${\cal S}$ defined in that theorem,
which includes the restrictions of both solutions to $[u_0,u_0+\tilde{U}]\times[0,v]$.
 Partitioning the rectangle $[0,u]\times[0,v]$ in a finite number
of strips of width at most $\tilde{U}$, uniqueness follows.
\end{proof}

An immediate consequence of Theorem~\ref{existence} and Proposition~\ref{uniqueness} is
\begin{Thm}[Maximal development and its domain $\mysigma$]\label{maximal} The characteristic initial value problem, with initial conditions\/
{\rm (I$_u$)} and\/ {\rm (I$_v$)} given for $u\in[0,\myumax]$ and $v\in[0,\infty[$, respectively, and with initial data satisfying {\rm (h1)$-$(h3)},
has a unique solution defined on a maximal past set $\mysigma$ containing a neighborhood of $[0,\myumax]\times\{0\}\cup\{0\}\times[0,\infty[$.\end{Thm}
\begin{proof}
Consider the union $\mysigma$ of the domains of all solutions of the characteristic initial value problem.
The set $\mysigma$ is a past set because it is a union of past sets and contains a neighborhood of $[0,\myumax]\times\{0\}\cup\{0\}\times[0,\infty[$
because of Theorem~\ref{existence} (see Figure~\ref{fig1_altb}).
Given $(u,v)\in\mysigma$ we define the solution on $(u,v)$ to be the solution on any past set containing $(u,v)$.
Due to Proposition~\ref{uniqueness} the solution is well defined.
\end{proof}

\begin{figure}[h!]
\begin{center}
\begin{turn}{45}
\psfrag{p}{\footnotesize $\cal{P}$}
\psfrag{u}{\tiny $v$}
\psfrag{0}{\tiny $0$}
\psfrag{v}{\tiny $u$}
\psfrag{h}{\tiny \!\!\!\!\!$U$}
\psfrag{l}{\tiny $r_0(u), \nu_0(u), \zeta_0(u)$}
\psfrag{n}{\tiny $\lambda_0(v), \varpi_0(v), \theta_0(v), \kappa_0(v)$}
\includegraphics[scale=.9]{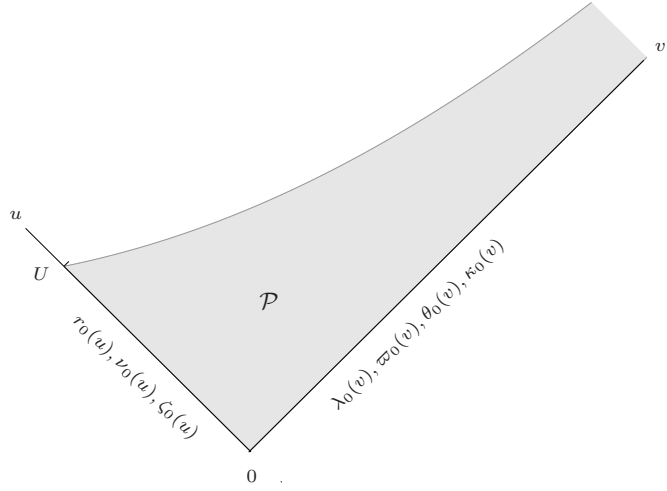}
\end{turn}
\end{center}
\vspace{-.5cm}
\caption{Domain of the maximal development.}\label{fig1_altb}
\end{figure}

\vspace{4mm}

The following version of Theorem~\ref{maximal} will be needed in Part~3.

\begin{Thm}[Maximal development for the backwards problem] The characteristic initial value problem
with initial conditions\/~{\rm (I$_u$)} given for $u\in]0,\myumax]$ and
$$({\rm I}^v)\qquad
\left\{
\begin{array}{lcl}
 \lambda(U,v)&=&\lambda_0(v),\\
 \omega(U,v)&=&\omega_0(v),\\
 \theta(U,v)&=&\theta_0(v),\\
 \kappa(U,v)&=&\kappa_0(v),
\end{array}
\right.\qquad{\rm for}\ v\in[0,V],
$$
satisfying\/ {\rm (h1)}$-${\rm (h3)},
has a unique solution defined on a maximal reflected past set\footnote{By reflected past set we mean a set $\mathcal{R}$ such that if $(u,v)\in \mathcal{R}$ then $[u,U]\times[0,v]\subset\mathcal{R}$.} $\mathcal{R}$ containing a neighborhood of $]0,\myumax]\times\{0\}\cup\{U\}\times[0,V]$.\end{Thm}
\begin{proof}
We just have to check that all the above holds if one substitutes the hypothesis $\nu_0(u)<0$ by the hypothesis $\nu_0(u)>0$, and then replace $u$ by $U-u$.
\end{proof}

In Part~3, we will also need
\begin{Lem}[Continuous dependence of the solution on the initial data]
Let $(r,\nu,\lambda,\varpi,\theta,\zeta,\kappa)$ and $(\tilde r,\tilde\nu,\tilde\lambda,\tilde\varpi,\tilde\theta,\tilde\zeta,\tilde\kappa)$
be the solutions of the characteristic initial value problem with initial data $(r_0,\nu_0,\lambda_0,\varpi_0,\theta_0,\zeta_0,\kappa_0)$
and $(\tilde r_0,\tilde\nu_0,\tilde\lambda_0,\tilde\varpi_0,\tilde\theta_0,\tilde\zeta_0,\tilde\kappa_0)$, respectively. Assume that the solutions are defined in $[0,U]\times[0,V]$. Furthermore, let
\begin{eqnarray*}
d(U,V)&=&(\|r-\tilde r\|+\|\nu-\tilde\nu\|+\|\lambda-\tilde\lambda\|+\|\varpi-\tilde\varpi\|\\
&&\ \ +\|\theta-\tilde\theta\|+\|\zeta-\tilde\zeta\|+\|\kappa-\tilde\kappa\|)_{L^\infty([0,U]\times[0,V])}.
\end{eqnarray*}
Then there exists a constant $C>0$, depending on $U$, $V$ and
\begin{eqnarray*}
D&:=&\Bigl(\|r\|+\|\nu\|+\|\lambda\|+\|\varpi\|
+\|\theta\|\\
&&\ \ 
+\|\zeta\|+\|\kappa\|+\left\|\frac1{r}\right\|+\left\|\frac1{\nu}\right\|\Bigr)_{L^\infty([0,U]\times[0,V])},
\end{eqnarray*}
such that, if
\begin{eqnarray*}\Delta_0&:=&(\|r_0-\tilde r_0\|+\|\nu_0-\tilde\nu_0\|+
\|\zeta_0-\tilde\zeta_0\|)_{L^\infty([0,U])}\\
&&\ \ +
(\|\lambda_0-\tilde\lambda_0\|+
\|\varpi_0-\tilde\varpi_0\|+\|\theta_0-\tilde\theta_0\|
+\|\kappa_0-\tilde\kappa_0\|)_{L^\infty([0,V])}
\end{eqnarray*}
is sufficiently small, then
\begin{eqnarray}
d(U,V)&\leq&C\Delta_0. \label{cdcd}
\end{eqnarray}
\end{Lem}
\begin{proof}
 Using~\eqref{kappa_final}$-$\eqref{r_final},
\begin{eqnarray*}
\|\varpi-\tilde\varpi\|_{L^\infty([0,U]\times[0,V])}&\leq& C\|\varpi_0-\tilde\varpi_0\|_{L^\infty([0,V])}+C(\|r-\tilde r\|\\
&&+\|\nu-\tilde\nu\|+\|\zeta-\tilde\zeta\|)_{L^\infty([0,U]\times[0,V])},\\
\|\kappa-\tilde\kappa\|_{L^\infty([0,U]\times[0,V])}&\leq& C\|\kappa_0-\tilde\kappa_0\|_{L^\infty([0,V])}+C(\|r-\tilde r\|\\
&&+\|\nu-\tilde\nu\|+\|\zeta-\tilde\zeta\|)_{L^\infty([0,U]\times[0,V])},\\
\|\lambda-\tilde\lambda\|_{L^\infty([0,U]\times[0,V])}&\leq& 
\|\lambda_0-\tilde\lambda_0\|_{L^\infty([0,V])}+C(\|r-\tilde r\|+\|\nu-\tilde\nu\|\\
&&+\|\varpi-\tilde\varpi\|+\|\kappa-\tilde\kappa\|)_{L^\infty([0,U]\times[0,V])}\\
&\leq& (\|\lambda_0-\tilde\lambda_0\|+C\|\varpi_0-\tilde\varpi_0\|+C\|\kappa_0-\tilde\kappa_0\|)_{L^\infty([0,V])}\\
&&+C(\|r-\tilde r\|+\|\nu-\tilde\nu\|+\|\zeta-\tilde\zeta\|)_{L^\infty([0,U]\times[0,V])},\\
\|\theta-\tilde\theta\|_{L^\infty([0,U]\times[0,V])}&\leq& \|\theta_0-\tilde\theta_0\|_{L^\infty([0,V])}+
C(\|r-\tilde r\|\\ &&+\|\lambda-\tilde\lambda\|+\|\zeta-\tilde\zeta\|)_{L^\infty([0,U]\times[0,V])}\\
&\leq& (C\|\lambda_0-\tilde\lambda_0\|+C\|\varpi_0-\tilde\varpi_0\|\\
&&+\|\theta_0-\tilde\theta_0\|+C\|\kappa_0-\tilde\kappa_0\|
)_{L^\infty([0,V])}
\\ &&+
C(\|r-\tilde r\|+\|\nu-\tilde\nu\|+\|\zeta-\tilde\zeta\|)_{L^\infty([0,U]\times[0,V])},\\
\|r-\tilde r\|_{L^\infty([0,U]\times[0,V])}&\leq&
\|r_0-\tilde r_0\|_{L^\infty([0,U])}\\
&&+
\int_0^V\|\lambda-\tilde\lambda\|_{L^\infty([0,U]\times[0,V'])}\,dV',\\
\|\nu-\tilde\nu\|_{L^\infty([0,U]\times[0,V])}&\leq& 
C\|\nu_0-\tilde\nu_0\|_{L^\infty([0,U])}+
C\int_0^V(\|r-\tilde r\|\\
&&\ +\|\varpi-\tilde\varpi\|+\|\kappa-\tilde\kappa\|)_{L^\infty([0,U]\times[0,V'])}\,dV',\\
\|\zeta-\tilde\zeta\|_{L^\infty([0,U]\times[0,V])}&\leq& \|\zeta_0-\tilde\zeta_0\|_{L^\infty([0,U])}+C\int_0^V(\|r-\tilde r\|\\ 
&&\ +\|\nu-\tilde\nu\|+\|\theta-\tilde\theta\|)_{L^\infty([0,U]\times[0,V'])}\,dV'.
\end{eqnarray*}
Combining the above inequalities, we get
\begin{eqnarray*}
d(U,V)&\leq&C(\|r_0-\tilde r_0\|+\|\nu_0-\tilde\nu_0\|+
\|\zeta_0-\tilde\zeta_0\|)_{L^\infty([0,U])}\\
&&+C
(\|\lambda_0-\tilde\lambda_0\|+
\|\varpi_0-\tilde\varpi_0\|\\
&&\ \ \ \ +\|\theta_0-\tilde\theta_0\|
+\|\kappa_0-\tilde\kappa_0\|)_{L^\infty([0,V])}\\
&&+C\int_0^Vd(U,V')\,dV'.
\end{eqnarray*}
Using Gronwall's inequality, we obtain~\eqref{cdcd}. The constant $C$ in~\eqref{cdcd} depends a priori on $D$ and
\begin{eqnarray*}
\tilde D&:=&\Bigl(\|\tilde r\|+\|\tilde\nu\|+\|\tilde\lambda\|+\|\tilde\varpi\|
+\|\tilde\theta\|\\
&&\ \ +\|\tilde\zeta\|+\|\tilde\kappa\|+\left\|\frac1{\tilde r}\right\|+\left\|\frac1{\tilde\nu}\right\|\Bigr)_{L^\infty([0,U]\times[0,V])}.
\end{eqnarray*}
However, the proof of Theorem~\ref{existence} shows that the norm of a solution is controlled by its initial data.
Using this fact together with~\eqref{cdcd}, we conclude that
if $\Delta_0$ is small, then
$\tilde D$ is controlled by, say, $D+1$. Therefore,
 the constant $C$ can be chosen depending only on $D$, as long as
$\Delta_0$ is sufficiently small.
The proof of the lemma is complete.
\end{proof}

With regard to the sign and monotonicity of the functions we can state
\begin{Lem}[Sign and monotonicity]\label{sign} 
Suppose that $(r,\nu,\lambda,\omega,\theta,\zeta,\kappa)$ is the solution of the characteristic initial value problem, with initial data satisfying\/ {\rm (h1)$-$(h3)}.
Then:
\begin{itemize}
 \item $\kappa$ is positive;
\item $\nu$ is negative;
\item $r$ is decreasing with $u$;
\item $\omega$ is nondecreasing with $v$;
\item if $\lambda(u,v)$ is negative (respectively, nonpositive), then $\lambda(u,v')$ is negative (respectively, nonpositive) for $v'>v$.
\end{itemize}
\end{Lem}
\begin{proof}
Using Definition~\ref{solution} and the sign of the initial data,
the function $\kappa$ is positive and the function $\nu$ is negative. From~(\ref{r_u}), the function $r$ is decreasing with $u$.
Then (\ref{omega_v})
shows $\omega$ is nondecreasing with $v$.

 If $\lambda(u,v)$ is negative, then since $\kappa$ is positive, equality (\ref{kappa_at_u}) shows that
$(1-\mu)(u,v)$ is negative (recall the definition of $\mu$ in~\eqref{mu_0}). On the other hand, using equations (\ref{r_v}), (\ref{nu_v}), (\ref{omega_v}) and (\ref{kappa_at_u}), we get
\begin{equation}\label{Ray} 
\partial_v\left(\frac{1-\mu}{\nu}\right)=-\,\frac{\theta^2}{\nu r\kappa}\geq 0.
\end{equation}
Therefore $(1-\mu)(u,v')$ is negative for $v'>v$.
From (\ref{kappa_at_u}), $\lambda(u,v')$ is negative for $v'>v$. The assertion regarding nonpositive $\lambda$ is proved in the same way.
\end{proof}

\begin{Rmk}
Notice that \eqref{kappa_u} and \eqref{Ray} recover the Raychaudhuri equations \eqref{ray_u} and \eqref{ray_v} from the solution of the first order system. 
\end{Rmk}

\section{Criterion for breakdown}

We now proceed to establish a breakdown criterion for the solution, beginning with the following crucial lemma.

\begin{Lem}[$r \to 0$ if and only if $\varpi\to +\infty$]\label{blow_up_b_semSinal}
Suppose that $(r,\nu,\lambda,\omega,\theta,\zeta,\kappa)$ is the solution of the characteristic initial value problem on $[0,U[\times[0,V[$
with initial data satisfying\/ {\rm (h1)$-$(h3)}.
Assume that $r$ is bounded. Then
$$
\inf_{(u,v)\in[0,U[\times[0,V[} r(u,v)=0\Leftrightarrow \sup_{(u,v)\in[0,U[\times[0,V[} \omega(u,v)=+\infty\;.
$$
\end{Lem}
\begin{proof}
Suppose that $(u,v)\in[0,U[\times[0,V[$ is such that
$$
 r(u,v)<\min_{u'\in[0,U]} r(u',0).
$$
Then there exists $0<v_*<v$ such that $\lambda(u,v_*)<0$. From Lemma~\ref{sign},
 we must have
$\lambda(u,v)< 0$.
Suppose there exists a sequence $(u_n,v_n)$ such that $r(u_n,v_n)\rightarrow 0$. Because $r$ has a positive minimum on
$[0,U]\times\{0\}$,
then $\lambda(u_n,v_n)<0$, for all $n$ sufficiently large. Thus $1-\mu(u_n,v_n)<0$, implying
\begin{equation}\label{bananaPequena}
\omega(u_n,v_n)>\frac{r(u_n,v_n)}{2}+\frac{e^2}{2r(u_n,v_n)}-\frac{\Lambda}{6}r^3(u_n,v_n)\rightarrow +\infty.
\end{equation}
Therefore, if $r$ goes to zero then $\omega$ goes to infinity.

To prove the converse, we will show that if there
exist constants $c_r$ and $C_r$ such that
\begin{equation}
\label{rBounds}
 0<c_r\leq r\leq C_r\;,
\end{equation}
in the rectangle $[0,U[\times[0,V[$, then $\omega$ is bounded in that rectangle. So $\omega$ cannot go to infinity without $r$ going to zero (recall that we are assuming that $r$ is bounded).

We start by showing that $\Aw$ can be chosen large enough so that
\begin{enumerate}[(i)]
\item $|\omega|<\Aw$\;,  on the initial segments  $[0,U]\times\{0\}\cup\{0\}\times [0,V]$;
\item $f(r,\omega):=\frac{2}{r^2}\frac{\vlinha}{\muu}$ is  bounded in absolute value by a constant $C_f$ in the region
$\{(u,v)\in[0,U[\times[0,V[:\omega(u,v)\geq \Aw\}$;
\item $\omega(u,v)\geq \Aw\  \Rightarrow\   \lambda(u,v)<0$.
\end{enumerate}
It is obvious that we can choose $A_\omega$ satisfying (i). By increasing $A_\omega$, if necessary, we can
guarantee (ii) because
$$ 
\lim_{\omega\rightarrow+\infty}f(r,\omega)=\frac{1}{r}\leq\frac{1}{c_r}.
$$ 
Finally, \eqref{mu_0} and \eqref{rBounds} show that if $A_\omega$ is sufficiently large and $\omega(u,v)>A_\omega$, then
$1-\mu(u,v)$ is negative.  For such an $A_\omega$, (iii) holds.

Let $\hat v:[0,U[\rightarrow [0,V]$ be defined by
\begin{equation}\label{defHatv}
\hat v(u):=\left\{
\begin{array}{l}
 \min\{v\in [0,V[ \, :\, \omega(u,v)=\Aw\}\; \text{ if the set is nonempty},\\ \\
 V \; \text{ otherwise}.
\end{array}
\right.
\end{equation}
We partition $[0,U[\times[0,V[$ in
$$
{\cal R_-}=\{(u,v)\in[0,U[\times[0,V[:0\leq v\leq \hat{v}(u)\}
$$
and
$$
{\cal R_+}=\{(u,v)\in[0,U[\times[0,V[:v\geq \hat{v}(u)\}
$$
(see Figure~\ref{fig2_altb}).

\begin{figure}[h!]
\begin{center}
\begin{turn}{45}
\begin{psfrags}
\psfrag{g}{\tiny $\cal{R}_-$}
\psfrag{j}{\tiny $\cal{R}_+$}
\psfrag{v}{\tiny $u$}
\psfrag{u}{\tiny $v$}
\psfrag{h}{\tiny $v=\hat{v}(u)$}
\psfrag{(0,0)}{$(0,0)$}
\psfrag{x}{\tiny \,$U$}
\psfrag{b}{\tiny \!$V$}
\includegraphics[scale=.8]{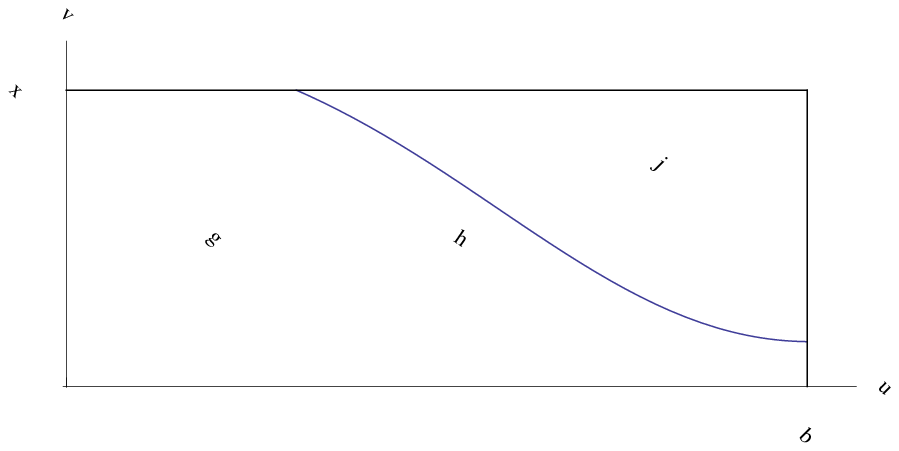}
\end{psfrags}
\end{turn}
\end{center}
\vspace{-.5cm}
\caption{The regions ${\cal R_-}$ and ${\cal R_+}$.}\label{fig2_altb}
\end{figure}

By the monotonicity of $\omega$ established in Lemma~\ref{sign},
we have
$$|\omega(u,v)| \leq \Aw\  {\rm if}\ (u,v)\in{\cal R}_- $$
and
\begin{equation}\label{quase}
\omega(u,v) \geq \Aw\  {\rm if}\ (u,v)\in{\cal R}_+.
\end{equation}

By \eqref{kappa_final}, we have $0<\kappa(u,v)\leq\kappa_0(v)$. On ${\cal R}_-$,
there exists a constant $C$ such that
\begin{equation}\label{nereida}
\left|2\kappa\frac{1}{r^2}\left(\vlinha\right)\right|\leq C.
\end{equation}
From \eqref{nu_final} and \eqref{nereida}, we have
\begin{equation}
\label{boundsNu}
 0<-\nu_0(u)e^{-CV}\leq -\nu(u,v)\leq -\nu_0(u)e^{CV}.
\end{equation}
On ${\cal R}_+$, we
 have $1-\mu<0$ and so we can use \eqref{kappa_at_u} and \eqref{nu_final} and the definition of $f$ in (ii) to write
$$\nu(u,v)=\nu(u,\hat v(u))\exp{\left(-\int_{\hat v(u)}^v\lambda(u,v') f(r(u,v'),\omega(u,v'))\,dv'\right)}\;.$$
Using (ii), (iii) and~\eqref{rBounds} we see that
\begin{eqnarray*}
\int_{\hat v(u)}^v|\lambda(u,v')| |f(r(u,v'),\omega(u,v'))|\,dv'
&\leq& -C_f\int_{\hat v(u)}^v \lambda(u,v')\,dv'
\\
&=& -C_f[r(u,v)-r(u,\hat v(u))]
\\
&\leq&  C_fC_r.
\end{eqnarray*}
This, together with~\eqref{boundsNu}, leads to
\begin{equation}
 \label{boundsNu2}
 0<-\nu_0(u)e^{-CV}e^{-C_fC_r}\leq -\nu(u,v)\leq -\nu_0(u)e^{CV}e^{C_fC_r}\ {\rm in}\ {\cal R}_+.
\end{equation}
The inequalities \eqref{boundsNu} and \eqref{boundsNu2} show that $\nu$ is bounded and bounded away from zero in $[0,U[\times[0,V[$.

We now wish to prove that $\lambda$  is bounded. Similarly to the definition of $\hat{v}$ in \eqref{defHatv},
we define $\hat u:[0,V[\rightarrow [0,U]$ by $$
\hat u(v):=\left\{
\begin{array}{l}
 \min\{u\in [0,U[ \, :\, \omega(u,v)=\Aw\}\; \text{ if the set is nonempty},\\ \\
 U \; \text{ otherwise}.
\end{array}
\right.
$$
Fix $v\in[0,V[$.
Let $u_1\in[\hat u(v), U[$ be such that $\omega(u_1,v)\geq \Aw$. Then, from (iii), there exists $\delta_1>0$ such that in $I_1:=[u_1,u_1+\delta_1[$ we have
\begin{eqnarray*}
\lambda(u,v)<0
&\Rightarrow&
1-\mu(u,v) <0
\\
&\Rightarrow&
\partial_u\omega(u,v)\geq 0\qquad({\rm by}\ \eqref{omega_u})
\\
&\Rightarrow&
\omega(u,v)\geq \Aw\;, \text{ for all } u\in I_1\;.
\end{eqnarray*}
From this it follows that
\begin{equation}
\label{omegaGrande}
 \omega(u,v) \geq \Aw\  {\rm if}\ u\in[\hat u(v),U[.
\end{equation}
Together \eqref{quase} and \eqref{omegaGrande} imply ${\cal R_+}$ is a future set and ${\cal R_-}$ is a past set, and so we may rewrite
$$
{\cal R_-}=\{(u,v)\in[0,U[\times[0,V[:0\leq u\leq \hat{u}(v)\}
$$
and
$$
{\cal R_+}=\{(u,v)\in[0,U[\times[0,V[:u\geq \hat{u}(v)\}.
$$
 From~\eqref{lambda_final}, \eqref{nereida} and the fact that $\nu$ is bounded,
the function $\lambda$ is bounded on ${\cal R_-}$. On ${\cal R_+}$ we have $1-\mu<0$ and so we are allowed to
substitute (\ref{kappa_at_u}) in (\ref{lambda_u}). This leads to
$$ 
\partial_u\lambda=\lambda\left(-\,\frac{2\nu}{1-\mu}\frac{1}{r^2}\left(\vlinha\right)\right)=\lambda\big(-\nu f(r,\omega)\big)
$$ 
and
$$
\lambda(u,v)=\lambda(\hat u(v),v)\exp\left\{-\int_{\hat u(v)}^u \nu(u',v) f(r(u',v),\omega(u',v))\,du'\right\}.
$$
Since $\nu$ is bounded, $f$ is bounded on ${\cal R_+}$, and $\lambda$ is bounded on ${\cal R_-}$, we conclude that
$\lambda$ is bounded on ${\cal R_+}$.

Now we combine (\ref{theta_final}) and (\ref{zeta_final}) to bound the pair
$(\theta,\zeta)$. We obtain
\begin{eqnarray*}
\theta(u,v)&=&\theta_0(v)-\int_0^u\zeta_0(u')\left(\frac{\lambda}{r}\right)(u',v)\,du'\\
&&+\int_0^u\left(\frac{\lambda}{r}\right)(u',v)\int_0^v\left(\frac{\theta\nu}{r}\right)(u',v')\,dv'du'.
\end{eqnarray*}
Let $0<\eps<V$
and
$\Theta_\eps(u)=\max_{v\in[0,V-\eps]}|\theta(u,v)|$. We have
\begin{eqnarray*}
\Theta_\eps(u)&\leq&\Theta_0+C+C\int_0^u\Theta_\eps(u')\,du',
\end{eqnarray*}
where $\Theta_0=\max_{v\in[0,V]}|\theta_0(v)|$ and $C$ is a positive constant independent of $\eps$.
Using Gronwall's inequality,
$$
\Theta_\eps(u)\leq (\Theta_0+C)e^{Cu}\leq (\Theta_0+C)e^{CU}.
$$
This shows that $\sup_{(u,v)\in[0,U[\times[0,V[}|\theta(u,v)|$ is finite and a similar computation shows that
$\sup_{(u,v)\in[0,U[\times[0,V[}|\zeta(u,v)|$ is finite.

Since  $\nu$ is bounded away from zero, \eqref{kappa_final} shows $\kappa$ is also bounded away from zero.
Integrating~\eqref{omega_v}, we conclude that
 $\omega$ is bounded.
\end{proof}

\begin{Rmk}
The proof that $r\to 0$ implies $\varpi \to +\infty$ requires $e \neq 0$ and $\kappa > 0$. The Schwarzschild black hole and the region $r<r_-$ of the Reissner-Nordstr\"{o}m solution, where $r \to 0$ but $\varpi$ is constant, show that these conditions are indeed necessary. The boundedness of $r$ is only needed in the proof of the converse implication; this hypothesis cannot be removed, as shown by one of the Wyman solutions\/ {\rm \cite{Wyman}} for the Einstein-scalar field equations (with $e=0$, but charge does not play any role here):
\[
g = - r^2 dt^2 + 2dr^2 + r^2 \sigma_{\mathbb{S}^2}, \qquad \phi = t.
\]
For this metric $\varpi = \frac{r}4$ tends to $+ \infty$ as $r \to +\infty$.
\end{Rmk}

The proof of the previous lemma also yields
\begin{Lem}[If $0<c_r\leq r\leq C_r$, then $\nu$, $\lambda$, $\omega$, $\theta$, $\zeta$ and $\kappa$ are bounded]\label{blow_up_a}
Under the conditions of the previous lemma, suppose
that $r$ is bounded below by a strictly positive constant $c_r$ (and above by a constant $C_r$). Then
$\nu$, $\lambda$, $\omega$, $\theta$, $\zeta$ and $\kappa$ are bounded in $[0,U[\times[0,V[$. Moreover,
$\nu$ and $\kappa$ are  bounded away from zero.
The bounds depend only on $c_r$ and $C_r$ (as well as on the initial data, on $U$ and on $V$).
\end{Lem}

Using Lemmas~\ref{blow_up_b_semSinal} and~\ref{blow_up_a}, we can prove
\begin{Thm}[Breakdown criterion]\label{breakdown}
Suppose that $(r,\nu,\lambda,\omega,\theta,\zeta,\kappa)$ is the maximal solution of the characteristic initial value problem with
initial data satisfying\/ {\rm (h1)$-$(h3)}, as in\/ {\rm Theorem \ref{maximal}}.
If $(\myU,V')$ is a point on the boundary of $\mysigma$ with $0<\myU<\myumax$ and $V'>0$, then
for all sequences $(u_n,v_n)$ in $\mysigma$ converging to $(\myU,V')$, we have
$$
r(u_n,v_n)\to 0\quad {\rm and}\quad \omega(u_n,v_n)\to\infty.
$$
\end{Thm}
\begin{proof}
According to \eqref{bananaPequena},  $${\rm if}\ r(u_n,v_n)\to 0,\ {\rm then}\ \omega(u_n,v_n)\to\infty.$$
Suppose, then, that the statement is not true. Then, after possibly extracting a subsequence, there exists a sequence $(u_n,v_n)$ in $\mysigma$ converging to $(\myU,V')$ such that
$r(u_n,v_n)$ is bounded below by a strictly positive constant.
Without loss of generality, we may assume $v_n\leq V'+1$.
We now show $$\min_{(u,v)\in[0,u_n]\times[0,v_n]}r(u,v)=\min_{v\in[0,v_n]}r(u_n,v)=\min\left\{r(u_n,0),r(u_n,v_n)\right\}.$$
Indeed, $r$ is decreasing with $u$.
If there existed $v'\in]0,v_n[$ with
\begin{equation}\label{below}
r(u_n,v')<\min\left\{r(u_n,0),r(u_n,v_n)\right\},
\end{equation}
then, by the mean value theorem, $\lambda(u_n,v'')<0$ for some $0<v''<v'$. By Lemma~\ref{sign}, this would imply $\lambda(u_n,v)<0$ for $v>v''$ and
hence $r(u_n,v')>r(u_n,v_n)$, contradicting \eqref{below}.
As
$$
r(u_n,0)\geq\min_{u'\in[0,\myumax]}r(u',0),
$$
we thus have uniform bounds (in $n$) for $r$ from below. Also, $r$ is uniformly bounded above by $\max_{v\in[0,V'+1]}r(0,v)$.

Fix a value of $n$. From the uniform boundedness of $r$ and Lemma~\ref{blow_up_a},
 we also have uniform bounds
for $\nu$, $\lambda$, $\omega$, $\theta$, $\zeta$ and $\kappa$
on $[0,u_n]\times[0,v_n]$.
By Theorem~\ref{existence},
there exist $\tilde{U}$ and $\tilde{V}$, independent of $n$, such that
the set $[0,u_n+\tilde{U}]\times[0,v_n]\cup [0,u_n]\times[0,v_n+\tilde{V}]$
is contained in ${\cal P}$, with uniform bounds on the solutions.
Applying again Theorem~\ref{existence}, now with initial data on $[u_n,u_n+\tilde{U}]\times\{v_n\}\cup
\{u_n\}\times[v_n,v_n+\tilde{V}]$, we see that for some positive $\delta$ independent of $n$
the region $[u_n,u_n+\delta]\times[v_n,v_n+\delta]$ is contained in ${\cal P}$.
For sufficiently large $n$, the domain of existence of the solution would contain $(\myU,V')$, thus contradicting the fact that $(\myU,V')$ belongs to the boundary of $\mysigma$.
The proof is complete.
\end{proof}

Alternatively, the previous result can be viewed as an extendibility condition in the following
way.

\begin{Cor}[Extendibility condition]\label{prolonga} 
Suppose $(r,\nu,\lambda,\omega,\theta,\zeta,\kappa)$ is the solution of the characteristic initial value problem on $[0,\myU[\times[0,V'[$ with $0<\myU<\myumax$
with initial data satisfying\/ {\rm (h1)$-$(h3)}.
If
$r$ is bounded above and below, then there exists a positive $\delta$ such that the solution can be extended to
$[0,\myU+\delta]\times[0,V'+\delta]$.
\end{Cor}

\begin{Rmk}
A rectangle in the region $r_+ < r < r_c$ of the Reissner-Nordstr\"om-de Sitter, with open edges on the event horizon $r=r_+$ and the cosmological horizon $r=r_c$, would appear to be a counter-example to the extendibility condition above, as the singularity prevents the extension; however, it does not satisfy the hypothesis $\kappa_0(v)>0$ in\/ {\rm (h2)}, since the null segments where the initial data is prescribed must necessarily extend past the event and cosmological horizons.
\end{Rmk}

\section{Derivation of the Einstein equations from the first order system}\label{Miranda}

We assume now the additional regularity condition
\[
\text{(h4)} \qquad \nu_0, \kappa_0 \text{ and } \lambda_0 \text{ are continuously differentiable.}
\]
Note that since $\lambda_0$ is continuous, $\tilde{r}_0$ is continuously differentiable. By (\ref{kappa_at_zero}), if $\kappa_0$
is continuously differentiable, then $\lambda_0$ is continuously differentiable. Conversely, if $\lambda_0$ is continuously differentiable
and the term inside the brackets in (\ref{kappa_at_zero}) is different from zero, then $\kappa_0$ is continuously differentiable.
\begin{Lem}[Hypotheses (h1)$-$(h4) imply that $r$ is $C^2$]\label{regular} 
Suppose that $(r,\nu,\lambda,\omega,\theta,\zeta,\kappa)$ is the solution of the characteristic initial value problem
with initial data satisfying\/ {\rm (h1)} to\/ {\rm (h4)}.
Then the function $r$ is $C^2$.
\end{Lem}
\begin{proof}
Since $\partial_u\partial_v r=\partial_u\lambda$ and $\partial_v\partial_u r=\partial_v\nu$, the mixed second derivatives of $r$ are continuous.
Assume $\kappa_0$ and $\lambda_0$ are continuously differentiable. From~(\ref{kappa_final}), $\partial_v k$ exists and is continuous.
Then (\ref{lambda_final}) implies $\partial_v\lambda$ exists and is continuous. Therefore, $\partial_v\partial_v r$ exists and is continuous.
On the other hand, assume $\nu_0$ is continuously differentiable. Then (\ref{nu_final}) implies $\partial_u\nu$ exists and is continuous.
Therefore, $\partial_u\partial_u r$ exists and is continuous.
\end{proof}

The type of arguments used in the proof of the previous lemma allow us to establish 
\begin{Prop}[Regularity of the solution of the first order system]\ \label{regularity}
\begin{enumerate}[{\rm (i)}]
\item If the initial data satisfy\/ {\rm (h1)} to\/ {\rm (h4)}, that is, $r_0\in C^2$, $\nu_0$, $\lambda_0$, $\varpi_0$ and $\kappa_0\in C^1$, and $\theta_0$ and $\zeta_0\in C^0$, then
$r\in C^2$, $\nu$, $\lambda$, $\varpi$ and $\kappa\in C^1$, and $\theta$ and $\zeta\in C^0$.
\item If $r_0\in C^3$, $\nu_0$, $\lambda_0$, $\varpi_0$ and $\kappa_0\in C^2$, and $\theta_0$ and $\zeta_0\in C^1$, then
$r\in C^3$, $\nu$, $\lambda$, $\varpi$ and $\kappa\in C^2$, and $\theta$ and $\zeta\in C^1$.
\end{enumerate}
\end{Prop}

\begin{proof}\
\begin{enumerate}[{\rm (i)}]
\item Taking into account the proof of Lemma~\ref{regular}, we just have to notice that $\partial_u\varpi$, $\partial_v\varpi$ and  $\partial_u\kappa$
are continuous according to~\eqref{omega_u}, \eqref{omega_v} and~\eqref{kappa_u}.
\item We just provide a sketch of the proof. We start by checking that $\partial_v\kappa$ is $C^1$.
The expression for $\partial_v\kappa$ involves $\kappa_0$, $\kappa_0'$, $r$, $\nu$, $\zeta$, $\partial_v\nu$ (which in turn involves  $r$, $\nu$, $\varpi$ and $\kappa$), $\lambda$ and $\partial_v\zeta$ 
(which in turn involves  $r$, $\nu$ and $\theta$).
These functions have continuous derivatives with respect to $u$ and $v$.
Indeed, $\kappa_0$ is $C^2$,
$\partial_u\nu$ is obtained from~\eqref{nu_final} using the fact that $\nu_0$ is $C^1$,
$\partial_u\zeta$ is obtained from~\eqref{zeta_final} using the fact that $\zeta_0$ is $C^1$,
$\partial_v\lambda$ is obtained from~\eqref{lambda_final} using the fact that $\lambda_0$ is $C^1$, and
$\partial_v\theta$ is obtained from~\eqref{theta_final} using the fact that $\theta_0$ is $C^1$.
The expression for $\partial_u\kappa$ involves $r$, $\nu$, $\zeta$ and $\kappa$. These functions have continuous derivatives with respect to $u$ and $v$.
This shows that $\kappa$ is $C^2$. 

Using the fact that $\kappa$ is $C^2$, one shows in a similar way that $\partial_u\nu$, $\partial_v\nu$, $\partial_u\lambda$ and $\partial_v\lambda$ are $C^1$.
This guarantees that $r$ is $C^3$. The remaining assertions follow easily.
\end{enumerate}
\end{proof}

\vspace{4mm}

Using \eqref{theta_u} and \eqref{zeta_v}, one immediately checks that
$$
\partial_u\Bigl(\frac\theta r\Bigr)=\partial_v\Bigl(\frac\zeta r\Bigr).
$$
Since ${\cal P}$ is simply connected, there exists a scalar field $\phi$ such that
\begin{eqnarray} 
 \theta&=&r\partial_v\phi,\label{jolie}\\
 \zeta&=&r\partial_u\phi.\label{angelina}
\end{eqnarray}
In terms of the variables of our first order system, the metric coefficient $\Omega^2$ is given by
$$
\Omega^2=-\,\frac{4\nu\lambda}{1-\mu}=-4\nu\kappa.
$$
This is equation~\eqref{omega_sq}. Recall that this equation is equivalent to~\eqref{bar_rafaeli}.

\begin{Prop}\label{einstein-2} 
 The functions $r$, $\phi$ and $\Omega$ satisfy the Einstein equations~\eqref{wave_r}, \eqref{wave_phi}, \eqref{r_uu} and~\eqref{r_vv}.
\end{Prop}
\begin{proof}
Writing (\ref{theta_u}) in terms of $\phi$, using~\eqref{r_v}, we obtain
$$
\partial_u\theta=-\partial_v r\,\partial_u\phi.
$$
On the other hand,
using~\eqref{jolie}, we get
$$
\partial_u\theta=\partial_ur\,\partial_v\phi+r\partial_u\partial_v\phi.
$$
Combining these two expressions for $\partial_u\theta$, we arrive at the wave equation for $\phi$:
$$
\partial_u\partial_v\phi=-\,\frac{\partial_ur\,\partial_v\phi+\partial_vr\,\partial_u\phi}{r}.
$$

In terms of $r$ and $\Omega$, both equation (\ref{lambda_u}) and (\ref{nu_v}) become the wave equation for $r$:
$$
\partial_u\partial_vr=\frac{\Omega^2}{2}\frac{1}{r^2}\left(\vlinha\right).
$$
Using~\eqref{kappa_u}, \eqref{Ray}, \eqref{jolie} and \eqref{angelina}, we get
$$
\partial_u\left(\frac{\partial_ur}{\Omega^2}\right)=-\,\frac 14\partial_u\left(\frac{1-\mu}{\lambda}\right)=\frac{\zeta^2}{4r\nu\lambda}(1-\mu)=-r\frac{(\partial_u\phi)^2}{\Omega^2}
$$
and
$$
\partial_v\left(\frac{\partial_vr}{\Omega^2}\right)=-\,\frac 14\partial_u\left(\frac{1-\mu}{\nu}\right)=\frac{\theta^2}{4r\nu\lambda}(1-\mu)=-r\frac{(\partial_v\phi)^2}{\Omega^2}.
$$
\end{proof}

\begin{Prop}\label{einstein-3} 
The first order system~\eqref{r_u}$-$\eqref{kappa_at_u} implies the wave equation~\eqref{wave_Omega} for $\ln\Omega$.
\end{Prop}
\begin{proof}
We start by
differentiating both sides of~\eqref{omega_sq} with respect to $u$. Using~\eqref{kappa_u} and then~\eqref{omega_sq},
\begin{eqnarray}
2\Omega\partial_u\Omega&=&-4\partial_u\nu\kappa-4\nu\left[\kappa\nu\frac 1r\left(\frac{\zeta}{\nu}\right)^2\right] \nonumber\\
&=&\Omega^2\left(\frac{\partial_u\nu}{\nu}+\frac{\zeta^2}{r\nu}\right). \label{egg_art}
\end{eqnarray}
Equation (\ref{nu_v}), written in terms of $\Omega^2$, becomes
\begin{equation}
\partial_v\nu=\frac{\Omega^2}{2}\frac{1}{r^2}\left(\vlinha\right).
\label{nu_v_altalt}
\end{equation}
Differentiating both sides with respect to $u$, using \eqref{egg_art}, we obtain
\begin{eqnarray}
\partial_u\partial_v\nu&=&\partial_u\left(\frac{\Omega^2}{2}\frac{1}{r^2}\left(\vlinha\right)\right)\nonumber\\
&=&\frac{\Omega^2}{2}\left(\frac{\partial_u\nu}{\nu}+\frac{\zeta^2}{r\nu}\right)
\frac{1}{r^2}\left(\vlinha\right)\nonumber\\
&&-\frac{\Omega^2}{r^3}\nu\left(\vlinha\right)\nonumber\\
&&+\frac{\Omega^2}{2r^2}\left(-\,\frac{e^2}{r^2}\nu+\Lambda r^2\nu+\frac{2\lambda\zeta^2}{\Omega^2}\right). \label{compenetration}
\end{eqnarray}
This shows that $\partial_u\partial_v\nu$ exists and is continuous, and so, since $\nu$ is $C^1$, $\partial_v\partial_u\nu = \partial_u\partial_v\nu$ also exists.
So, using~\eqref{zeta_v}, \eqref{jolie}, \eqref{angelina} and \eqref{egg_art}
\begin{eqnarray}
&&\partial_v\partial_u\ln\Omega=\partial_v\left(\frac{\partial_u\Omega}{\Omega}\right)\\
&&\qquad=\frac{1}{2}\partial_v\left(\frac{\partial_u\nu}{\nu}
+\frac{\zeta^2}{r\nu}\right)\nonumber\\
&&\qquad=\frac{1}{2}\left(\frac{\partial_v\partial_u\nu}{\nu}-\,\frac{\partial_u\nu\partial_v\nu}{\nu^2}+\frac{2\zeta\partial_v\zeta}{r\nu}
-\,\frac{\zeta^2}{r^2\nu}\lambda-\,\frac{\zeta^2}{r\nu^2}\partial_v\nu\right)\nonumber\\
&&\qquad=\frac{1}{2}\left(\frac{\partial_v\partial_u\nu}{\nu}-\,\frac{\partial_u\nu\partial_v\nu}{\nu^2}-\frac{2\zeta\theta}{r^2}
-\,\frac{\zeta^2}{r^2\nu}\lambda-\,\frac{\zeta^2}{r\nu^2}\partial_v\nu\right)\nonumber\\
&&\qquad=\frac{1}{2}\left(\frac{\partial_v\partial_u\nu}{\nu}-\,\frac{\partial_u\nu\partial_v\nu}{\nu^2}-2\partial_u\phi\,\partial_v\phi
-\,\frac{\zeta^2}{r^2\nu}\lambda-\,\frac{\zeta^2}{r\nu^2}\partial_v\nu\right).\label{O_uv}
\end{eqnarray}
We now replace \eqref{nu_v_altalt} and \eqref{compenetration} in \eqref{O_uv}.
After doing so, we are left with two terms that involve $\partial_u\nu$, which add up to zero,
and four other terms that involve $\zeta^2$, which also add up to zero:
\begin{eqnarray*}
\frac{1}{2}\left(\frac{\Omega^2}{2}\frac{\zeta^2}{r^3\nu^2}\left(\vlinha\right)+\frac{\lambda\zeta^2}{r^2\nu}\right.\ \ \ \ \ \ &&\\
\left.-\,\frac{\Omega^2}{2}\frac{\zeta^2}{r^3\nu^2}\left(\vlinha\right)-\,\frac{\lambda\zeta^2}{r^2\nu}\right)&=&0.
\end{eqnarray*}
Besides the term $-\partial_u\phi\partial_v\phi$, there remain three terms in the final expression. Their sum is
\begin{equation}\label{sum_three}
\frac{\Omega^2}{2}\left(-\,\frac{1}{r^3}\left(\vlinha\right)-\,\frac{e^2}{2r^4}+\frac{\Lambda}{2}\right).
\end{equation}
Replacing formula~\eqref{bar_rafaeli} for $\omega$ in (\ref{sum_three}), we get
$$
\frac{\Omega^2}{2}\left(-\,\frac{e^2}{r^4}-\frac{\Lambda}{3}+\frac{e^2}{2r^4}
+\frac{1}{2r^2}-\frac{\Lambda}{6}+\frac{2}{\Omega^2r^2}\nu\lambda
-\,\frac{e^2}{2r^4}+\frac{\Lambda}{2}\right).
$$
This simplifies to
$$
\frac{\Omega^2}{2}\left(-\,\frac{e^2}{r^4}
+\frac{1}{2r^2}+\frac{2}{\Omega^2r^2}\nu\lambda\right)=-\,\frac{\Omega^2e^2}{2r^4}+\frac{\Omega^2}{4r^2}+\frac{\partial_ur\,\partial_vr}{r^2},
$$
and so we obtain equation~\eqref{wave_Omega}.
\end{proof}

\begin{Rmk}
Since equations~\eqref{wave_r}$-$\eqref{r_vv} imply the first order system,
equations~\eqref{wave_r}$-$\eqref{r_vv} also imply~\eqref{wave_Omega}.
\end{Rmk}

Regularity for the solution of the first order system implies regularity of the metric $g$ and of the field $\phi$:
\begin{Rmk}[Regularity of the metric and the field]\ \label{rmk-reg}
\begin{enumerate}[{\rm (i)}]
\item
In the case of\/~{\rm Proposition~\ref{regularity}\,(i)}, the metric~$g$ is~$C^1$ and the field~$\phi$  is~$C^1$. 
\item
In the case of\/~{\rm Proposition~\ref{regularity}\,(ii)}, the metric~$g$ is~$C^2$ and the field~$\phi$ is~$C^2$. 
\end{enumerate}
\end{Rmk}
One can easily generalize Proposition~\ref{regularity} and Remark~\ref{rmk-reg} to higher orders of regularity.
Note that all these results also hold for the backwards problem.

\end{document}